\documentclass[12pt]{article}
\usepackage{url}
\usepackage{pdflscape}
\usepackage{rotating}
\usepackage{graphicx,bbm}
\usepackage{natbib}
\usepackage{amsfonts, animate,subfigure}
\usepackage{amsmath,amssymb,amsthm,soul}
\usepackage[margin=1in]{geometry}
\usepackage{setspace}
\usepackage[export]{adjustbox}
\usepackage[rightbars, color]{changebar}
\cbcolor{blue}
\usepackage{soul}

\newcommand{\floor}[1]{\lfloor #1 \rfloor}
\usepackage[usenames,dvipsnames,table]{xcolor}

\usepackage{thmtools}
\usepackage{thm-restate}

\usepackage[normalem]{ulem}

\newcommand{\mb}[1]{\mathbf{#1}}
\newcommand{\bs}[1]{\boldsymbol{#1}}
\newcommand{\s}[1]{\mathcal{#1}}
\newcommand{\mbb}[1]{\mathbb{#1}}

\newcommand{\n}[1]{\mathrm{#1}}
\title{\LARGE Multiresolution network models\\[10pt]}
\author{Bailey K. Fosdick\footnote{Authors listed in alphabetical order.  Contact emails: bailey.fosdick@colostate.edu and tylermc@uw.edu.}\\{\small Colorado State University} \and Tyler H. McCormick\\{\small University of Washington} \and Thomas Brendan Murphy\\{\small University College Dublin} \and Tin Lok James Ng\\{\small University College Dublin} \and Ted Westling\\{\small University of Washington}}
\date{}

\begin{document}

\maketitle

\begin{abstract}
Many existing statistical and machine learning tools for social
network analysis focus on a single level of analysis.  Methods
designed for clustering optimize a global partition of the graph,
whereas projection based approaches (e.g. the latent space model in
the statistics literature) represent in rich detail the roles of
individuals.  Many pertinent questions
in sociology and economics, however, span multiple scales of analysis.  Further, many questions involve comparisons across disconnected graphs that will, inevitably  be of different sizes, either due to missing data or the inherent heterogeneity in real-world networks.  We propose a class of network models that represent network structure on multiple scales and facilitate comparison across graphs with different numbers of individuals.  These models differentially invest
modeling effort within subgraphs of high density, often termed
communities, while maintaining a parsimonious structure between said
subgraphs.  We show that our model class is projective, highlighting an ongoing
discussion in the social network modeling literature on the dependence of 
inference paradigms on the size of the observed graph. 
We illustrate the utility of our method using data on household relations from Karnataka, India.
\end{abstract}

\doublespacing

\section{Introduction}
Social network data consist of a sample of actors
and information on the presence/absence of pairwise relationships 
among them.  These data are often represented as a graph where nodes correspond to actors and edges (ties) connect nodes with a relationship.  A relationship may represent, for example, friendship between students, co-authorship between academics on a journal article, or a financial transaction between organizations. Understanding structure in social networks is essential to appreciating the nuances of human behavior and is an active area of research in the social sciences \citep{borgatti2009network}. Existing statistical models for social networks typically focus on either (i) carefully representing structure among actors 
that have a relatively high likelihood of interaction, or (ii) clearly differentiating between groups of actors, i.e. communities, within the graph that have high within-group connectivity and low between-group connectivity (see e.g.~\citet{salter2012review} for a review).  Unfortunately, neither of these approaches fully characterize the complexities displayed in many real-world social networks.  
 
 Observed graphs frequently exhibit a mixture structure that manifests through a combination of global sparsity and local density. Global sparsity implies that the propensity for a tie between any two randomly selected actors is incredibly small.  Yet, massive heterogeneity in the propensity for actors to connect often creates local graph structure concentrated in dense subgraphs, frequently termed \emph{communities}.  This structure is typically
 particularly pronounced in very large graphs.  For example, in the context of online communication networks,~\citet{ugander2011anatomy} describe the Facebook graph as containing pockets of ``surprisingly dense'' structure, though overall the graph is immensely sparse.  

In this paper, we propose a multiresolution model for capturing heterogeneous, complex structure in social networks that exhibit strong community structure. 
Our modeling framework decomposes network structure into a component that describes between-community relations, i.e. relations between actors belonging to different communities, 
and another component describing within-community relations.    
The proposed framework has two distinct advantages over existing methods.  First, our framework is able to 
accommodate a wide variety of models for between- and within-community relations.  This feature allows the model to be tailored to reflect different scientific questions that arise when exploring the behavior within and across these communities. The second advantage of our model is that it balances parsimony with model richness by selectively directing modeling efforts towards representing interesting, relevant network structure.  Typically, this structure is found within actors' local communities.  In such cases, we can exert the most modeling effort (i.e. model complexity and computational effort) within dense pockets, where we expect the most complex dependence structure, and use a parsimonious model to capture between community patterns.  A similar approach has been adopted in spatial statistics where locations are partitioned into disjoint dependence neighborhoods \citep{page2016spatial}. Compared to popular network models that capture global structure, our approach 
can provide increased resolution on intricate structure within communities.  Furthermore, our model is able to apportion little effort to modeling simple structure, resulting in a model that is substantially less complex than existing models focused on local structure for networks, even with only a few hundred actors.

After defining our model framework, we discuss its statistical properties.  In exploring these properties, we take a traditional sampling perspective and consider our observed network as that pertaining to a collection of actors sampled from an infinite population of actors.  Our goal is to learn features of the infinite population from the observed graph.  
In our model, these features include the distribution of 
within-community structure across the network. Communities are defined by their structure and may not have a consistent size.  Therefore, inference about the population-level parameters requires that we be able to coherently compare and summarize parameters associated with subgraphs of different sizes.  We may also desire to compare  network-level parameters to those from another network of a different size and a different number of communities.  In order for these properties to hold and comparisons to be meaningful, the model class must be a \emph{projective family}, in the ``consistency under sampling'' sense of~\citet{shalizi2013consistency}. We introduce this concept and show the class of multiresolution network models proposed have this property.  We also discuss the important implications of this for population inference.  

In the remainder of this section, we explore the two existing approaches to multiresolution modeling of networks, highlighting their strength and weaknesses.  In  
Section~\ref{sec:multires}, we introduce the general form of our multiresolution modeling framework, and in Section~\ref{sec:LS-SBM} we present one of many possible model instantiations, called the Latent Space Stochastic Blockmodel (LS-SBM). Section~\ref{sec:properties} describes the projectivity properties of this framework and provides context through comparison with other available methods. Finally, we conclude with a discussion in Section~\ref{sec:discuss}.

\subsection{Related models}

\label{sec:related}
In this section, we describe two existing models, the Latent Position Cluster Model~\citep{handcock2007model} and the Locally Dependent Exponential Random Graph Model~\citep{schweinberger2014local}, which  capture aspects of our multiresolution approach.  

The Latent Position Cluster Model (LPCM) of \citet{handcock2007model} is an extension of the latent geometry framework introduced in \citet{HRH02}, where the probability of network ties is a function of the distance between actor positions in a latent space.  The LPCM performs model-based clustering~\citep{fraley2002model} on the positions in the unobserved social space.  Cluster memberships then capture group structure and within-group analysis is performed by examining the actor-specific latent positions within each cluster.  The likelihood for the LPCM, like that for the original latent space models, requires estimating a distance between every two pair of actors in the unobserved social space.  In even moderately large 
graphs, these distance calculations are computationally expensive and the propensity for actors in different groups to interact is often very small.  In addition, since both ties and non-ties are weighted equally by the LPCM, the latent position for each node is heavily influenced by the numerous other nodes with which it has no relation.

An attractive property of the LPCM is that it parsimoniously encodes patterns among ties in the network using a low-dimensional structure. As a consequence of this, the model 
tie probabilities are constrained by the
latent geometry.  Often these constraints are seen as model features; for example, the triangle inequality encourages transitivity, which is known to be prevalent in empirical social networks.  However, these constraints also restrict the types of network structure that can be represented. 
Consider a two-dimensional Euclidean latent space and suppose there are four groups of actors 
such that each pair of actors which belong to different groups interact with the same probability. To model this type of between group structure in the latent space, all groups need to be positioned equidistant from one another. Unfortunately, it is impossible to place four points in $\mathbb{R}^2$ equidistant from one another.  The dimension of the latent space could be increased to accommodate this structure.  However, ultimately a $K-1$ dimensional space is required to model all possible relationships among $K$ groups and continually increasing the latent dimension greatly complicates the model.

The second recent and related model is the locally dependent Exponential Random Graph Model proposed in \citet{schweinberger2014local}.  Exponential Random Graph Models (ERGMs) (e.g. \citet{frankmarkov1986,wasserman1996logit,pattison1999logit,snijders2002markov,robins2007introduction, koskinen2009, robins2011exponential,chatterjee2013estimating}) use
graph 
statistics, such as the number of closed triangles, as sufficient statistics in an exponential family. 
\citet{schweinberger2014local} define local dependence on a graph as a decoupling of the dependence structure such that dependence exists only among ties within the same community and among ties between the same two communities. 
~\citet{schweinberger2014local} draw parallels to local (dependence) neighborhoods in spatial statistics and M-dependence in time series.  

A key feature of locally dependent ERGMs is that they are composed of ERGMs.  \citet{shalizi2013consistency} showed that ERGMs are not projective when the sufficient statistics involve more than two nodes. Lacking projectivity implies that the value of the model parameters changes meaning depending on the sample size.  As a result, when partial communities are observed (that is, some of the actors in various communities are not included in the sample), the parameter estimates from locally dependent ERGMs are difficult to interpret.
Further, since the ties within each community are modeled using an ERGM, it is not possible to compare parameters across communities within the same graph unless the communities happen 
to be the same size.  Locally dependent ERGMs are in fact projective if the sampling units 
are taken to be \emph{communities} rather than actors. \citet{schweinberger2014local} calls this limited form of projectivity \emph{domain consistency}.  
While our proposed model uses a similar decomposition across subgraphs as~\citet{schweinberger2014local}, we model within- and between-community structure using latent variable mixture models and show the model class we define is projective when indexed by actors, broadening the notion of local dependence and alleviating the challenges with interpretation and comparison.

\section{Multiresolution network model}
\label{sec:multires}
In this section, we propose a general modeling framework that reflects the global sparsity and local density, ``chain of islands''~\citep{cross2001knowing}, structure observed in many large networks.

Consider a hypothetical infinite population of actors and communities, where each actor is a member of a single community.  Define $\gamma: \mbb{N} \to \mbb{N}$ to be the community membership map, which partitions the actors into disjoint communities. That is, $\gamma(i) = \gamma_i$ is the community of actor $i$. Define $K_N = | \gamma(\{1, \dotsc, N\}) |$ to be the number of unique communities among actors $\{1,...,N\}$. 
The community map is only meaningful up to relabellings of the communities.  Without loss of generality we require that $\gamma_i \leq K_{i-1} + 1$ (defining $K_0 = 0$) so that actors $1, \dotsc, N$ span communities $1, \dotsc, K_N$.  In practice the community memberships are typically estimated from the data, though in certain circumstances it can be defined \emph{a priori} from known structural breaks in the network (see~\citet{sweet2013hierarchical} for an example). 

Let $S_k = \{i : \gamma_i = k\}$ be the collection of actors in community $k$ in the population. We assume that  the number of actors in each community, e.g. $|S_k|$, is bounded, implying that $K_N = O(N)$ as $N\to\infty$.  Our assumption of bounded communities is supported by empirical evidence that suggests the ``best" communities contain small sets of actors, which are almost disconnected from the rest of the network \citep{leskovec2009community} and by psychologists and primatologists who have proposed a limit on the size of human social networks (e.g. \citet{dunbar1998social}).  This fact allows us to strategically allocate modeling effort within a large graph to be concentrated on a relatively small portion of dyads.  Since we postulate that the structure of within-community relations will
typically be most complex and interesting, we desire a model flexible enough to differentially devote modeling effort to those relations.  

Restricting community sizes to be bounded is also consistent with~\citet{schweinberger2014local}, though their motivation is quite different. Following similar justification as in the time series and spatial contexts, ~\citet{schweinberger2014local} define a decomposition of the graph such that the propensity to form ties between any set of nodes depends only on a finite number of other nodes. Along with global sparsity, the finite communities assumption in~\citet{schweinberger2014local} 
facilitates 
their asymptotic normality results for graph statistics.      

The network ties among a sample of $N$ actors can be represented as an $N \times N$ symmetric matrix $\mbb{Y}_{N}$, with $(i,j)$ entry $y_{ij} \in \{0,1\}$ denoting the absence or presence of a tie between actors $i$ and $j$.  We focus on undirected relations, restricting $y_{ij} = y_{ji}$ and assume the relation between an actor and itself, $y_{ii}$, is undefined. The partition $\gamma$ of the actors then induces a partition of the network $\mbb{Y}_N$ into \emph{blocks} $\mbb{Y}_{N,kl} = \{Y_{ij} : 1 \leq i < j \leq N, i \in S_k, j \in S_l\}$. We call $\mbb{Y}_{N,kk}$ a \emph{within-community} block and $\mbb{Y}_{N, kl}$, where $k \neq l$, a  \emph{between-community} block.  We define \emph{multiresolution network models} as the class of distributions over $\mbb{Y}_N$ such that, for a specific vector of $\gamma$'s, each distribution in the class can be expressed
\begin{equation} 
\small
\hspace{-.25in} P_{\gamma, \alpha, \omega, N}(\mbb{Y}_N) = \prod_{k = 1}^{K_N} W_{\alpha}( \mbb{Y}_{N,kk}) \prod_{k= 1}^{K_N - 1} \prod_{l=k+1}^{K_N}B_{\omega}(\mbb{Y}_{N,kl}) \label{eq:multires}
 \end{equation}
where $W$ is the probability distribution depending on $\alpha$ associated with the within-community model, and $B$ is the  probability distribution depending on $\omega$ associated with the between-community model. The population parameters $\alpha$ and $\omega$ characterize the distribution of within-block and between-block structure, respectively.

Furthermore, we express the within-community and between-community probability distributions as mixture distributions
\begin{align}
W_{\alpha}( \mbb{Y}_{N,kk}) &= \int W( \mbb{Y}_{N,kk}|\eta_k)\, dR_{\alpha}(\eta_k) \label{eq:within} \\
B_{\omega}(\mbb{Y}_{N,kl}) &= \int B(\mbb{Y}_{N,kl}|\tau_{kl})\, dS_{\omega}(\tau_{kl})
\label{eq:between}
\end{align}
and require the functional form of $B_\omega(\cdot)$ and $W_\alpha(\cdot)$ do not depend on the size of the network, $N$. 
Thus $\eta_k$ is the within-community random effect for community $k$ with random effect distribution $R_{\alpha}$ and $\tau_{kl}$ is the between-community random effect for the pair of communities $k, l$ with distribution $S_{\omega}$. The dimension of both $\eta_k$ and $\tau_{kl}$ do not depend on the sizes of the respective blocks as they come from common distributions $R_\alpha$ and $S_\omega$.  We assume that both $W$ and $B$ are projective, which is in contrast to the~\citet{schweinberger2014local} strategy of assuming ERGMs as the between and within block distributions, which are generally not projective.  Projectivity of $W$ and $B$ is essential for coherent inference based on the model; the importance of this property is detailed in Section~\ref{sec:properties}. 

In our probabilistic framework, we assume communities are exchangeable. That is, we assume that the community labels can be arbitrarily permuted and the probability of the network remains unchanged, or equivalently that there is no information about the social structure of the network contained in the specific values of the community labels.  We also assume that the node labels within each community are exchangeable.  
A familiar case where exchangeability does not hold is network data collected via snowball sampling, where nodes are progressively sampled by following ties in the network and nodes close together in the sampling order are likely to be connected. 

We call the models in \eqref{eq:multires} \emph{multiresolution} because the model parameters and random effects correspond to parameters at three resolutions. At the coarsest (global) level, $\gamma$ defines the distribution of community sizes. Related literature on community detection defines  
groups based on subgraph densities, such that actors have a higher propensity to interact within the group than between groups~\citep{newman2006modularity}, or based on \emph{stochastic equivalence}, where groups include actors that display similar interaction patterns to the rest of the network \citep{lorrain1971structural}. The most well-known models for community identification is the stochastic block model (SBM) and its variants~\citep{nowicki2001estimation, airoldi2008mixed, rohe2011spectral,choi2012stochastic,amini2013pseudo}.  While these methods distinguish between clusters of actors and their aggregate structure at a macro-level, they lack the ability to encode low-level structure such as transitivity, which manifests as triangles in the network. Recent extensions of SBMs include ``multiscale'' versions (e.g.~\citet{peixoto2014hierarchical}, \citet{lyzinski2015community}), which repeatedly subdivide clusters.  These models also fall within our general class. 

At the local level, $\alpha$ represents a (multivariate) population parameter determining the distribution of within-community structure across the population, where we might expect the richest structure. Recall the discussion in Section \ref{sec:related} about the LPCM and the inherent constraints on the model tie propensities due to the latent space embedding.  
When multiple dense pockets exist in a large, overall sparse network, the LPCM can sacrifice accuracy in characterizing within cluster structure to distinguish between the clusters. 
~\citet{vivar2012models} document a case 
involving baboon interactions where the separation of the dense communities in a troop dominates the latent space positioning, 
crowding out within group structure.
Our model resolves this issue by disentangling the modeling of within and between group relations.  In particular, $\omega$ characterizes the distribution of between-community relational structure across the population, 
and hence controls the overall sparsity and small-world properties of 
the network. 

While $\gamma$, $\alpha$, and $\omega$ summarize global structure, at a finer resolution, $\eta_k$ and $\tau_{kl}$ represent the unique structure present within and between specific communities. Finally, at the finest level of resolution, any actor-specific latent variables apart of each within-community distribution $W$ or each between-community relation distribution $B$ provide local representations of any involute structure. In the next section, we provide more concrete examples of these parameters in the context of popular network models.

The general class of multiresolution models defined in \eqref{eq:multires} contains a diverse set of possible model specifications.  The 
stochastic blockmodel \citep{holland1983stochastic}, for example, is a special case.  The stochastic blockmodel decomposes a network into communities and models the probability of a tie between any two actors as  solely a function of their community memberships.  In our formulation, $\eta_k$ would denote the tie probability between two individuals within community $k$ and $\tau_{kl}$ would denote the tie probability between an actor in community $k$ and an actor in community $l$.  Viewing these block-level probabilities as random effects, $R_\alpha$ and $S_\omega$ represent the mixing distribution governing the distributions of these effects, and $W$ and $B$ represent products of independent and identically distributed Bernoulli distributions.  We could also construct a model that nests stochastic blockmodels within one-another \citep{peixoto2014hierarchical}.  With this approach, we would represent both within- and between-community structure as a stochastic blockmodel.  Greater nesting depths might be specified for the within-community stochastic blockmodels to capture more complex patterns within communities.  Furthermore, separate random effects for sender and receiver effects could be added to each block, as in the social relations model~\citep{kenny1984social}.  In the following section, we explore another example of our model class, which we call the Latent Space Stochastic Blockmodel.

\section{Latent Space Stochastic Blockmodel}
\label{sec:LS-SBM}
In this section, we introduce a particular multiresolution network model, called the Latent Space Stochastic Blockmodel (LS-SBM). In the LS-SBM, the propensity for within-block ties is modeled with a latent space model and the between-block ties are modeled as in a stochastic blockmodel.  We denote the probability an actor belongs to block $k$ as $\pi_k$, and let $\boldsymbol{\pi} = (\pi_1,...,\pi_{K_N})$ denote the vector of membership probabilities, where $\sum_{i=1}^{K_N} \pi_i = 1$.

The LS-SBM utilizes a latent Euclidean distance model~\citep{HRH02} for the within-community distribution $W$. In this model the edges in $\mbb{Y}_{N,kk}$ are conditionally independent given the  latent positions 
of the actors in $S_k$. Specifically, given $\mb{Z}_i$ and $\mb{Z}_j$ where $i,j \in S_k$, $Y_{ij}$ is Bernoulli with probability $\n{logit}^{-1}(\beta_k - \|\mb{Z}_i - \mb{Z}_j\|)$, where $\|\cdot\|$ denotes the $\ell^2$-norm, i.e. Euclidean distance. The latent positions in group $k$ are themselves independent and identically distributed (IID) as spherically normal with mean zero and variance $\sigma^2_k$:  $\; N_D(\mb{Z}_i; 0,\sigma_k^2 I_D)$. Thus
\begin{equation} W (\mbb{Y}_{N, kk} | \eta_k)= \int \left(\prod_{i,j} G(Y_{ij}; \beta_k, \mb{Z}_i, \mb{Z}_j)\right) \, \left(\prod_i dN_{D}(\mb{Z}_i ; 0, \sigma_k^2 I_D)\right), \end{equation}
where $G$ is the Bernoulli distribution stated above and the products are taken with respect to all nodes in block $k$.

In the terminology of multiresolution models, $\eta_k \equiv (\beta_k, \log\sigma_k)$ is the within-community random effect governing the network structure within community $k$. $\beta_k$ can be interpreted as the maximum logit-probability of a relation in block $k$: two nodes $i$ and $j$ are stochastically equivalent in block $k$ if and only if $\mb{Z}_i = \mb{Z}_j$, in which case $P(Y_{ij} | \mb{Z}_i = \mb{Z}_j) = \n{logit}^{-1}(\beta_k)$. $\sigma_k$ is a measure of heterogeneity in block $k$, as $\sigma_k = 0$ is equivalent to an Erd\H{o}s-Reny\'i model with tie probability $\n{logit}^{-1}(\beta_k)$.  If $\sigma_k = 0$ for all blocks, then the multiresolution model reduces to the stochastic blockmodel.  We model $\eta_k$ for all $k= 1, \dotsc, K_N$ as samples from a bivariate normal with parameters $\alpha  = \{\bs\mu, \bs\Sigma\}$.  Thus, $
R_{\alpha}(\eta_k) = N_2((\beta_k, \log\sigma_k); \bs\mu, \bs\Sigma)$.

For the between-community distribution $B$ we use an Erd\H{o}s-Reny\'i model. That is, all edges between communities $k$ and $l$ are IID with probability $\tau_{kl}$. Thus
\begin{equation} B( \mbb{Y}_{N, kl} | \tau_{kl}) = \prod_{i \in S_k} \prod_{j \in S_l} \tau_{kl}^{Y_{ij}} (1-\tau_{kl})^{1 - Y_{ij}}. \end{equation}
We model ${\tau}_{kl}$ as Beta distributed with parameters $\omega = (a_0, b_0)$: $
S_{\omega}(\tau_{kl}) =  \n{Beta}(\tau_{kl}; a_0, b_0).$
The model maintains a parsimonious structure in modeling relationships between blocks, requiring only a single parameter, but is flexible in modeling ties within each block, allowing tie prevalence to depend on the distance between actors in the unobserved social space.

There are two key differences between the proposed LS-SBM and the latent position cluster model (LPCM) introduced in \citet{handcock2007model}.  The first key distinction is that the probability of a tie between actors belonging to different communities in the LS-SBM is a function of only their community memberships, whereas in the LPCM it is a function of the distance between the actor positions in the latent space.  This means that in the LPCM, all of an actor's ties and non-ties are used in determining the latent positions.  In contrast, in the LS-SBM, the latent space only affects within-community connections.  As a result, the structure of between-community connections are not constrained by the dimension and geometry of the latent space in the LS-SBM like they are in the LPCM.  The second distinction between the models is that the LPCM contains a single intercept parameter and the LS-SBM contains block-specific intercepts, $\beta_k$.  Each intercept $\beta_k$ can be interpreted as the maximum logit-probability of a tie in community $k$.  In practice, we find there is often large heterogeneity in this maximum probability across communities, suggesting having different intercepts is a critical piece of model flexibility in the LS-SBM.

\subsection{Prior specification}
\label{sec:priors}

Here we discuss the prior distributions for $\alpha$ and $\omega$. Our intended application of the LS-SBM is to networks where the within-community ties are denser than the between-community ties. We use the prior on $\alpha$ to reflect this knowledge.  Specifically, we set the prior on $\bs\mu$ and $\bs\Sigma$ to be a conjugate Normal-Inverse-Wishart distribution, with parameters $\{\mb{m}_0, s_0, \bs\Psi_0, \nu_0\}$, subject to an additional  \emph{assortativity restriction}. Given $a_0$ and $b_0$, the prior can be expressed
\begin{align}
\hspace{-.1in}P(\bs\alpha | a_0, b_0, \mb{m}_0, s_0, \bs\Psi_0, \nu_0) \propto&  N_2(\bs\mu; \mb{m}_0, \bs\Sigma_0 / s_0) \n{Inv.Wish}(\bs\Sigma; \bs\Psi_0, \nu_0)  \boldsymbol{1}(a_0,b_0,\bs{\mu}), \label{eq:prior}
\end{align}
where $\boldsymbol{1}(a_0,b_0,\bs{\mu})$ is the indicator function enforcing the assortativity condition. We fix $a_0$ and $b_0$ based on the observed density of the graph. The assortativity condition we require is that the (logit) marginal probability of a within-community tie for the \emph{average block}, induced by $\bs\mu$, be larger than the (logit) average between-block probability of a tie, $a_0/(a_0 + b_0)$:
\begin{equation}
E\Big[\mathrm{logit}(Pr(Y_{ij}=1))|\gamma_i = \gamma_j\Big] \ge E\Big[\mathrm{logit}(Pr(Y_{ij}=1))|\gamma_i \not= \gamma_j\Big]. \label{eq:assort} 
\end{equation}
Calculating these expectations (see the web-based supplementary materials for details), the restriction on the population parameter space we wish to enforce is
$$\mu_1 -  2 e^{\mu_2} \frac{\Gamma(\frac{D+1}{2})}{\Gamma(\frac{D}{2})} \ge  \psi(a_0) - \psi(b_0).$$ 
where $\psi(x)$ is the digamma function defined $\psi(x) = \frac{d \text{log}(\Gamma(x))}{dx}$. The digamma function is not available in closed form but is easily approximated with most standard statistical software packages. We proceed with estimation using this global assortativity restriction.

 \subsection{Estimation and block number selection}
 \label{sec:fitting}

Here we provide a brief summary of our estimation procedure for the LS-SBM. A full description of the model specification and algorithm are provided in the web-based supplementary materials. 
The posterior for our Bayesian model is not available in closed form, so instead we approximate the posterior using draws obtained via Markov chain Monte Carlo (MCMC).  The MCMC algorithm performs the estimation with the number of blocks $K$ fixed.  Thus, we first describe a procedure for choosing $K$ and then outline the MCMC procedure given the number of blocks.

We suggest comparing different $K$ using a series of ten-fold cross-validation procedures.  
For each repetition of the procedure, randomly partition the unordered node pairs into ten folds. For each fold, use assortative spectral clustering~\citep{saade2014spectral} on the adjacency matrix, excluding that fold, to partition the nodes into numbers of blocks $K$ from $2$ to $\floor{N/4}$.  To adapt the spectral clustering algorithm to deal with the held out, missing at random, edges we propose using an iterative EM-like scheme where first the missing values are imputed using observed degrees, clustering is performed on the imputed data, and the missing values from hold-out are re-imputed using the predicted probabilities. This should be repeated, until convergence.  This procedure does not require computing the full posterior and can be done in parallel for each value of $K$ and each validation fold.

For each repetition and $K$ value, we propose calculating three metrics of predictive performance: area under the ROC curve (AUC), mean squared error (MSE), and mean predictive information (MPI). Calculate the mean value and 95\% CIs for the mean of these criteria over the repetitions for each $K$. Then, for each criteria, we suggest finding the \emph{smallest} $K$ such that the mean value of the criteria for $K$ falls in the 95\% CI of the mean value of the $K$ with the best mean value (either maximal or minimal depending on the criterion).  

Once we have selected a value of $K$, we use a Metropolis-within-Gibbs algorithm to approximate the joint posterior distribution. Our use of conjugate priors allows Gibbs updates of $\bs\pi, \bs{\tau}, \bs\mu$, and $\bs\Sigma$. We update each $\eta_k$ with a Metropolis step, using a bivariate normal proposal distribution. 

Each node is assigned a single block membership at every iteration of the chain. This block membership is jointly updated with the node's latent position using a Metropolis step. We also take additional (unsaved) Metropolis steps for the latent positions in order to allow nodes which have switched blocks to find higher likelihood points in the latent space.

The likelihood is invariant to permutations of the block memberships and to rotations and reflections of the latent spaces. We address these non-identifiabilities by  post-processing the posterior 
samples using equivalence classes defined over the parameter spaces. See the supplementary materials for additional details. 

In our experiments, computation was feasible for networks with three hundred nodes in under two hours using a personal computer with a 2GHz processor and 8GB of RAM.  The most computationally expensive piece is the Metropolis step that jointly updates block membership and latent positions and then subsequently takes additional draws from latent spaces.  In our experiments, however, using a joint update substantially improved mixing.  Since we do not use MCMC to compute $K$, this is not a limiting step computationally.  
To scale our method beyond what is possible with MCMC, or for cases where the full posterior is not of interest, we provide a two-stage fitting procedure in the supplementary material.  This procedure uses an assortative graph clustering algorithm  to quickly estimate block membership, and variational inference to estimate parameters within each block. We have used this two-stage procedure to estimate our model in a sparse network with 13,000 nodes, which took about three minutes on a standard personal computer. The results of this analysis are provided in Section~\ref{sec:results}.  Additional details about the two-stage procedure are provided in the supplementary material.

\subsection{Karnataka villages}
\label{sec:results}

We estimate the proposed LS-SBM on data from a social network study consisting of households in villages in Karnataka, India to illustrate the utility of the model.  These data were collected as part of an experiment to evaluate a micro-finance program performed by~\citet{banerjee2013diffusion}.  Data consist of multiple undirected relationships between individuals and households in 75 villages.  Relationship types include social and familial interactions (e.g. being related or attending temple together) and views related to economic activity (e.g. lending money or borrowing rice/ kerosene).  

We used the household-level ``visit" relation from village 59, which has $N=293$ households with non-zero degree. We estimate the LS-SBM on the data with $K=6$.  Details of the cross-validation selection procedure for $K$ and LS-SBM estimation on this data are provided in the supplementary materials. Codes to replicate the results we present here are available at \url{https://github.com/tedwestling/multiresolution_networks.git}.  Data are available at \url{https://dataverse.harvard.edu/dataset.xhtml?persistentId=hdl:1902.1/21538}.

\begin{figure}[t]
\centering
\includegraphics[width=2.25in]{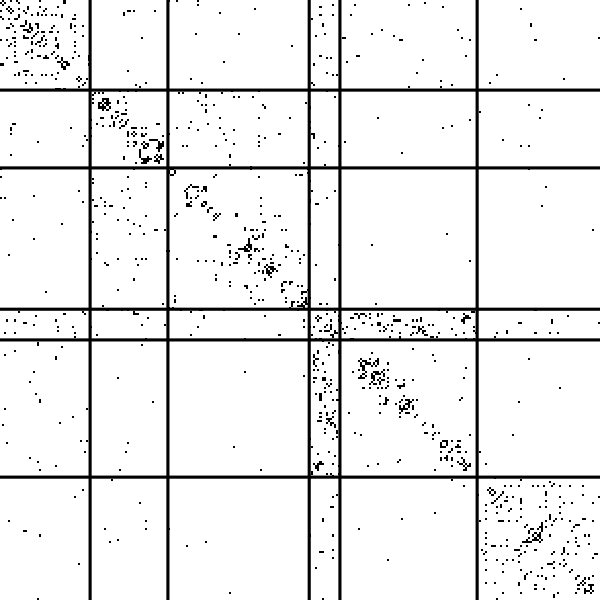}
\caption{Household-level ``visit" relation adjacency matrix of village number 59 from the Karnataka village dataset. Nodes are grouped by marginal posterior mode block membership. \label{fig:blocked_adjacency}}
\end{figure}

\begin{figure}[t]
\centering
\includegraphics[width=5.5in]{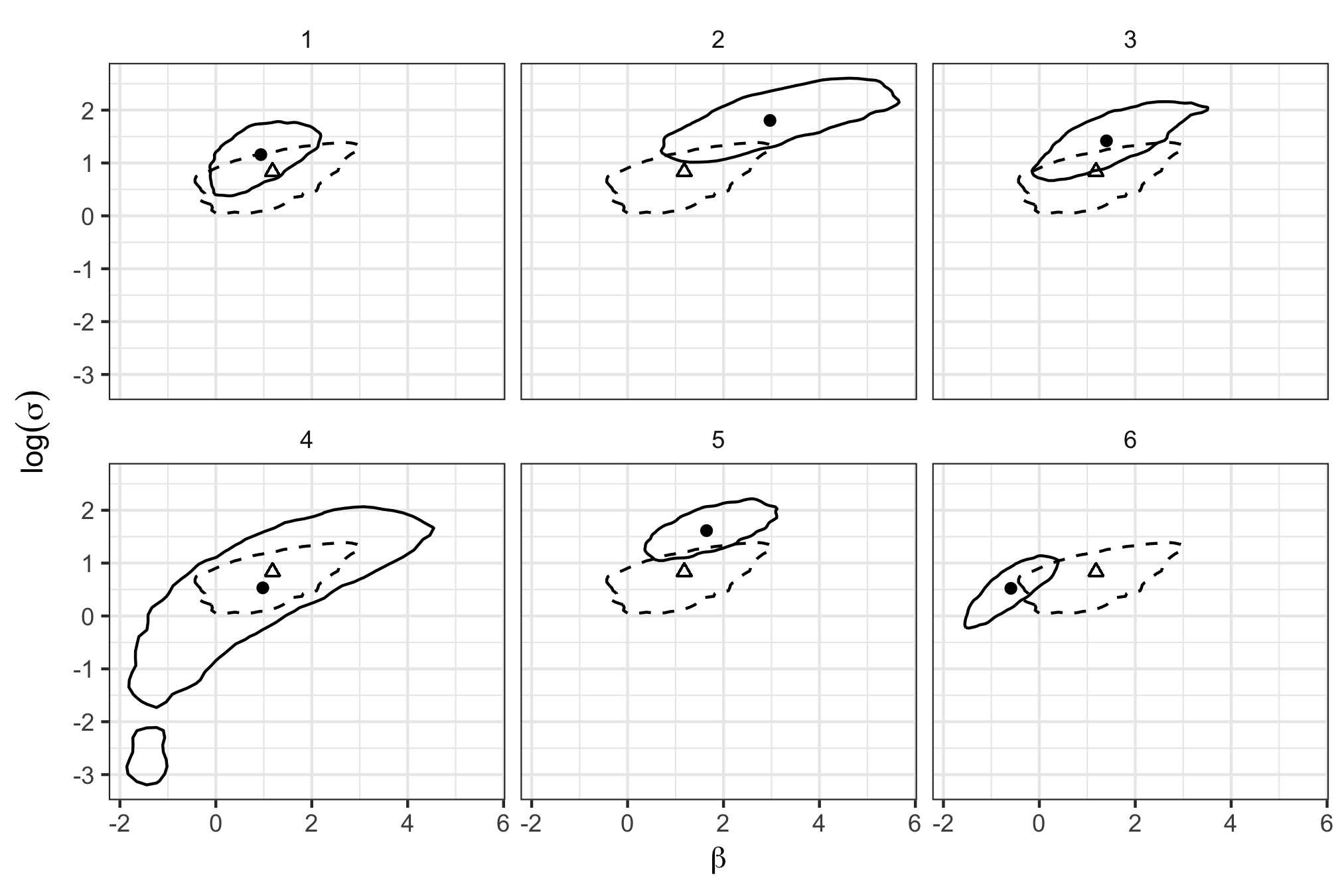}
\caption{Block-level and global parameter estimates and 95\% highest posterior density (HPD) regions. Each panel contains the posterior mean and 95\% HPD region for a single block parameter $(\beta, \log\sigma)$ (solid point and solid line) and the posterior mean and 95\% HPD of the global parameter $(\mu_1, \mu_2)$ (triangle and dashed line).}
\label{fig:block_parameters}
\end{figure}

\begin{figure}[t]
\centering
\includegraphics[width=6in]{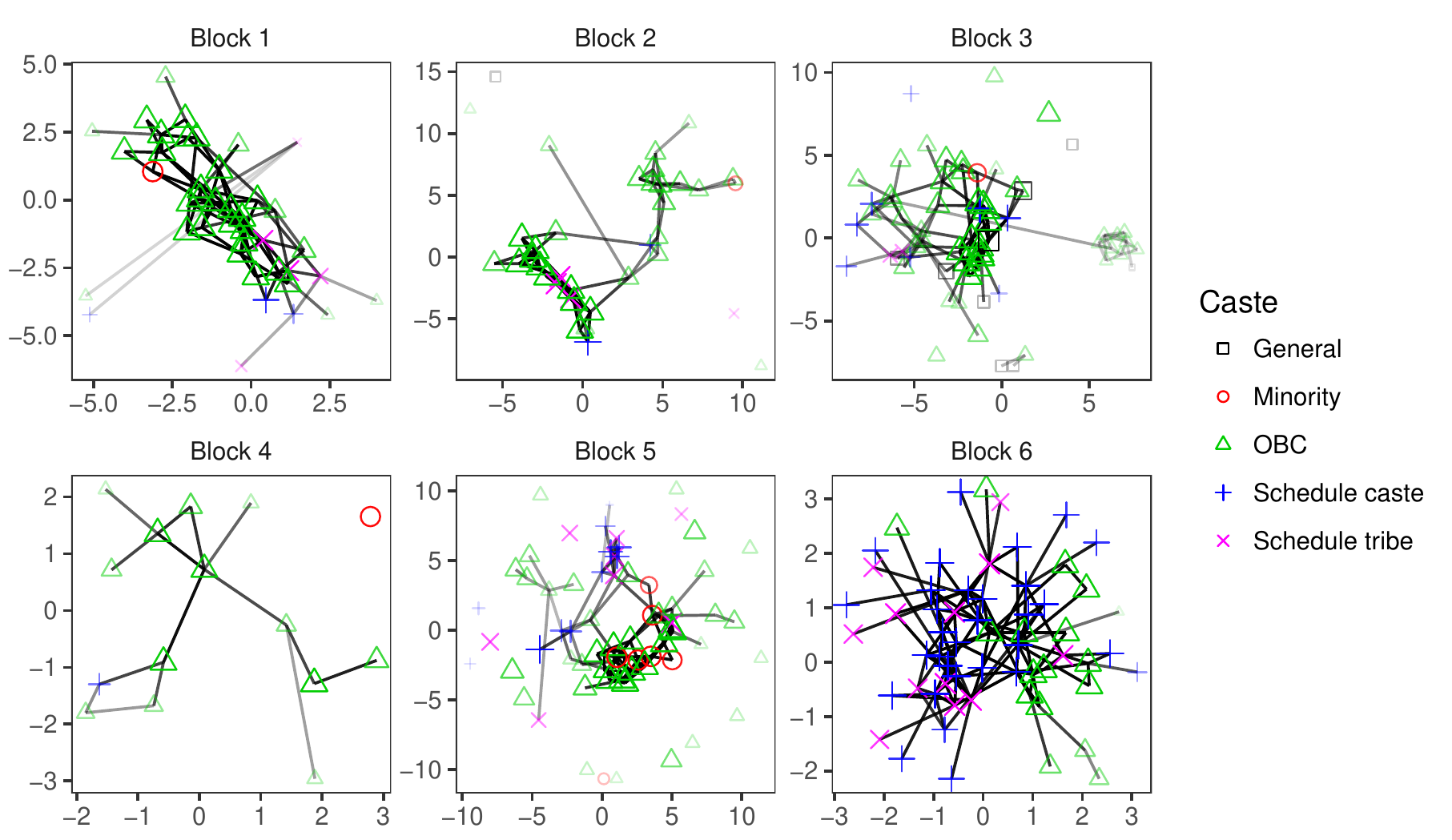}
\caption{Latent positions within each block.  Shading represents posterior probability of inclusion in a given block, with the lightest shading representing a posterior estimate of around 30\% and darkest colors representing values near unity. Colors and plot symbols differentiate the five caste categories.
\label{fig:latent_spaces}}
\end{figure}

The results of the estimation algorithm are displayed in Figures \ref{fig:blocked_adjacency}, \ref{fig:block_parameters}, and \ref{fig:latent_spaces}, and in Table \ref{tab:between_block_ests}. Figure~\ref{fig:blocked_adjacency} shows the observed adjacency matrix organized by marginal posterior mode block membership. The estimated block memberships result in an assortative network structure: there are more ties between households in the same community than between households in different communities. Table \ref{tab:between_block_ests} shows the posterior mean estimates of the between block connectivity parameters, again illustrating the assortative patterns. The within-community ties seen in Figure~\ref{fig:blocked_adjacency} are fairly clearly non-uniform, further justifying our departure from a SBM for the within-community model.

\begin{table}
\caption{Between block probability matrix.  The off-diagonal elements are the posterior mean probabilities of a tie between individuals in different blocks, $\tau_{\gamma_i\gamma_j}$.  The diagonal elements represent the maximum probability of a tie within each block based on the block-level parameter $\beta_k$ posterior mean: $e^{\hat{\beta}_k}/(1+e^{\hat{\beta}_k})$.  Values less than 0.01 are grayed out.}
\label{tab:between_block_ests}
\centering
\begin{tabular}{c||c|c|c|c|c|c}
     & 1 & 2 & 3 & 4 & 5 & 6\\ \hline\hline
     1& $\le$.719 &\cellcolor{lightgray}.006 &\cellcolor{lightgray}.004&.044 &.010 &\cellcolor{lightgray}.003 \\ \hline
     2& \cellcolor{lightgray}.006 &$\le$ .951&.018& .016&\cellcolor{lightgray}.006 &\cellcolor{lightgray}.003\\ \hline
     3&\cellcolor{lightgray}.004 &.018& $\le$ .802&.011&\cellcolor{lightgray}.003&\cellcolor{lightgray}.002\\ \hline
     4&.044 &.016&.011&$\le$ .727&.079&.021\\ \hline
     5& .010&\cellcolor{lightgray}.006&\cellcolor{lightgray}.003&.079&$\le$ .838&\cellcolor{lightgray}.002\\ \hline
     6&\cellcolor{lightgray}.003 &\cellcolor{lightgray}.003&\cellcolor{lightgray}.002&.021&\cellcolor{lightgray}.002& $\le$ .356\\ 
\end{tabular}
\end{table}

Figure~\ref{fig:block_parameters} shows the estimated block-level parameters $\eta_k = (\beta_k, \log \sigma_k)$ and the global mean $\boldsymbol{\mu}=(\mu_1, \mu_2)$.  Recall that $\beta_k$ is the intercept parameter for each block and $\log \sigma_k$ describes the variation in the latent space.  The $\mu_1$ and $\mu_2$ terms describe the mean of the distribution of $\beta_k$ and $\log \sigma_k$, respectively.  Since the multiresolution framework is projective, we can compare parameters between block-level parameters. Blocks two, five and six, denoted B2, B5 and B6, respectively, display strong \emph{a posteriori} differences from the mean block. In B2 the posterior distribution is shifted towards a larger intercept parameter, with posterior probability approximately 97\% that $\beta_2 > \mu_1$.  The larger intercept indicates that the overall propensity for ties is larger in B2.  In addition, the posterior probability that $\log \sigma_2$ is greater than the overall mean $\mu_2$ is greater than 99\%,  suggesting there is greater heterogeneity in tie probabilities in block two compared to the global mean. In B6, the $\beta_6$ term is shifted substantially lower than the overall mean $\mu_1$, indicating a smaller maximum tie probability.  The posterior probability that $\beta_6 < \mu_1$ is greater than 99\%.

We now further explore the structure implied within each block by examining the within-block latent positions.  Figure~\ref{fig:latent_spaces} shows the latent positions (after accounting for nonidentifiability issues) obtained from multidimensional scaling on the posterior mean distance matrix, along with the observed edges within each block.  The posterior places a distribution over block memberships for each node, however in Figure~\ref{fig:latent_spaces} we show nodes only in the block for which they have the largest posterior probability of inclusion.  Individuals that are unlikely to belong to any specific block (with inclusion probabilities less than 30\%) are omitted.  Shading represents the concentration of the posterior over block memberships, with darker colors indicating higher assignment probabilities. 

Moving to the structure within the blocks, we investigate the relation between the household memberships, positions, and caste which is a formalized social class system in India. Castes are denoted in Figure~\ref{fig:latent_spaces} using different colors.  We see strong sorting by caste, with almost all of the members of schedule castes and schedule tribes (the two lowest castes) being grouped into B6.  Members of the slightly higher class OBC (Other Backwards Caste) are the most common in the network and are spread throughout the remaining blocks.  An extensive literature in economics (e.g. \citet{townsend1994risk,munshi2009mobility, mazzocco2012testing, ambrus2014social}) explores on the role of the caste system in individuals' financial decisions.  In particular this literature focuses on informal credit markets.  That is, the social structures that provide financial support in times of need without a formal, corporate credit structure.  Recent work by~\citet{ambrus2014social} present a theoretical argument for the importance of ties that bridge otherwise disconnected groups.  In our results, these individuals would be individuals whose block assignment based on their social interactions does not match that of others in their caste.  For example, this group of bridging individuals would include members of schedule tribes or castes that are in blocks other than B6. 

We also used our two-stage procedure (described in detail in the supplementary material) to estimate our model for all 75 village networks combined. We formed an undirected network of $N=$ 13,009 nodes by combining all 75 household-level ``visit" relation networks from the Karnataka village data. We estimated the block structure using label propagation \citep{raghavan2007near}, which returned 534 blocks. Every block contained only households from a single village -- that is, there were no blocks containing households from multiple villages. This was expected since by design there are no between-village edges. There were a median of six blocks per village, with as few as one block per village (i.e.\ the entire village constitutes a single block) and as many as twenty blocks per village. The number of nodes per block varied considerably, with a median of fourteen, mean of 24.4, and maximum of 233. The density of edges within a block and nodes per block were well-described by a linear function on the log-log scale with intercept 0.34 and slope -0.82, as shown in the left panel of Figure~\ref{fig:all_village_plots}. The estimated within-block latent space parameters $\log\sigma$ and $\beta$ are shown in the right panel of Figure~\ref{fig:all_village_plots}. The larger blocks tend to be sparser and more heterogeneous, while the smaller blocks are more homogeneous. 

\begin{figure}[t]
\centering
\includegraphics[width=2.2in]{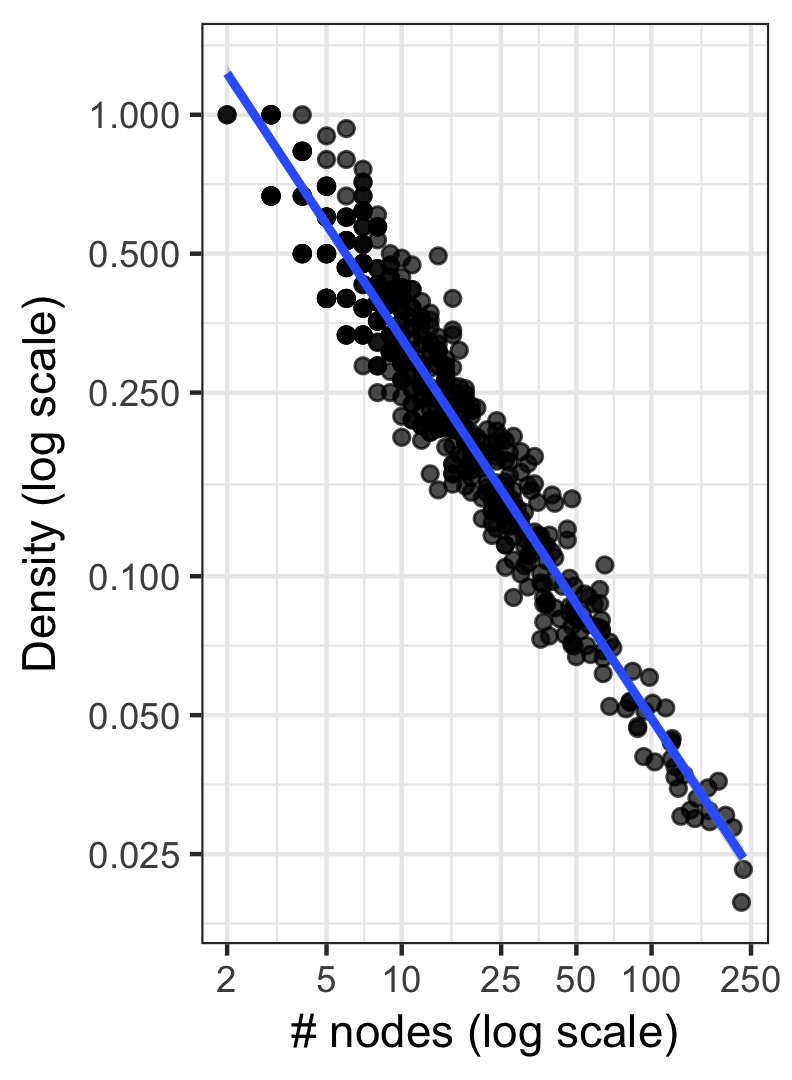}
\includegraphics[width=4.2in]{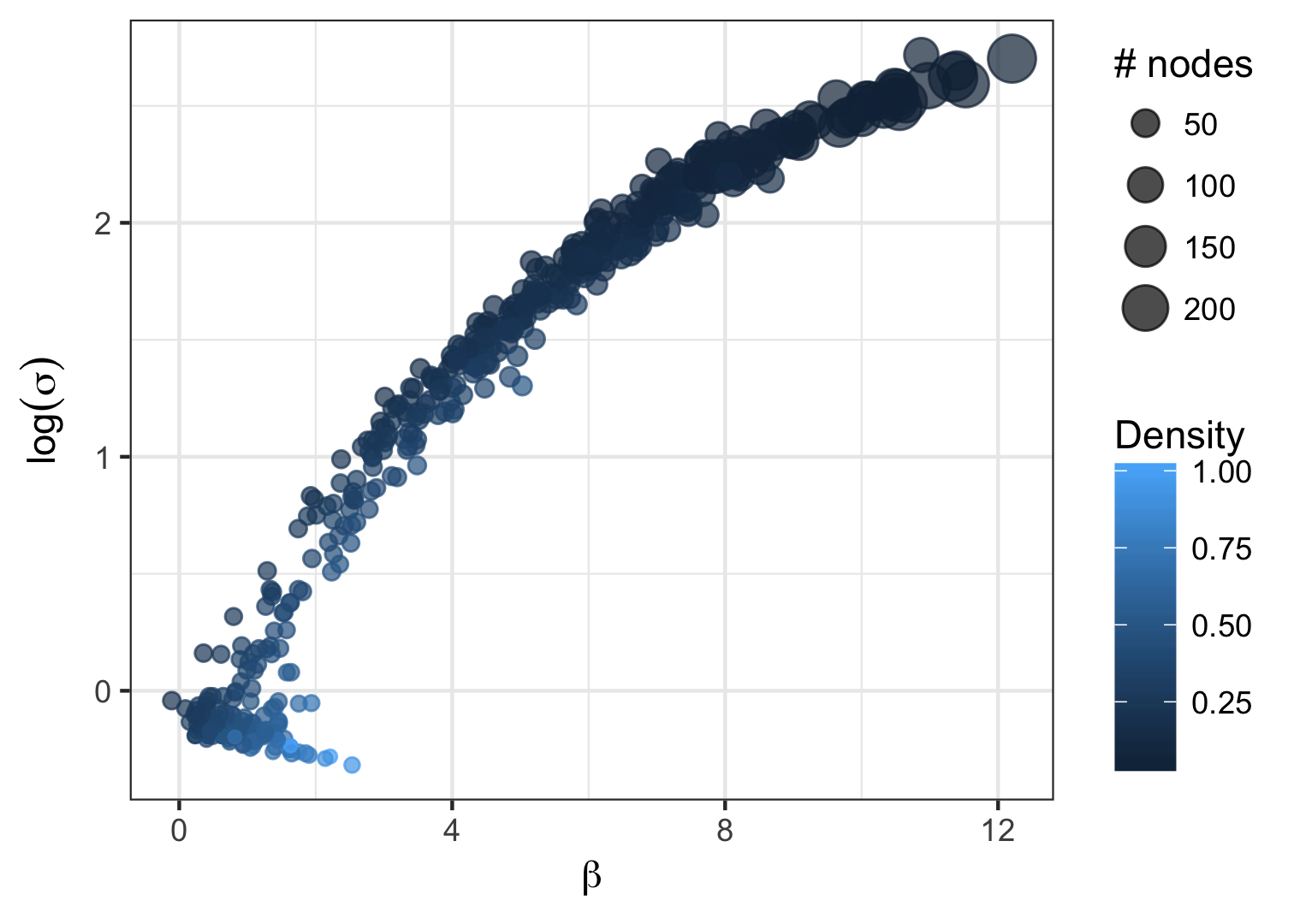}
\caption{The left panel shows the $\log_{10}$ block density as a function of $\log_{10}$ block size. The blue line is the OLS linear regression fit. The right panel shows the block-level latent space parameters $\beta_k$ and $\log(\sigma_k)$, where point size corresponds to block size and point color corresponds to block density. In both plots, each point is an estimated block from the model fit to all 75 Karnataka villages using the two-stage procedure.}
\label{fig:all_village_plots}
\end{figure}

\subsection{Simulation Study}

In this section, we detail a simulation study illustrating the advantages of using a multiresolution model like the LS-SBM over existing models such as the latent space model, LPCM and SBM.

Binary, undirected network data were generated for 300 nodes from the LS-SBM model with five equally-sized blocks such that the between-block tie probabilities were either $0.2$ or $0.02$. Within-block tie probabilities stemmed from a heterogeneous set of two-dimensional block-specific latent spaces.  Further details about the simulation parameters are provided in the supplementary materials. One thousand simulations were performed where ten percent of the undirected dyads in the network were held out in each simulation and the models were fit to the remaining ninety percent of the data.  Predictions were then made for the held out portion of the network and the accuracy of these predictions quantified by computing the area under the precision-recall curves.  The results are shown in Figure \ref{fig:sim}.

\begin{landscape}
\begin{figure}
    \centering
\includegraphics[width=1.2\textwidth]{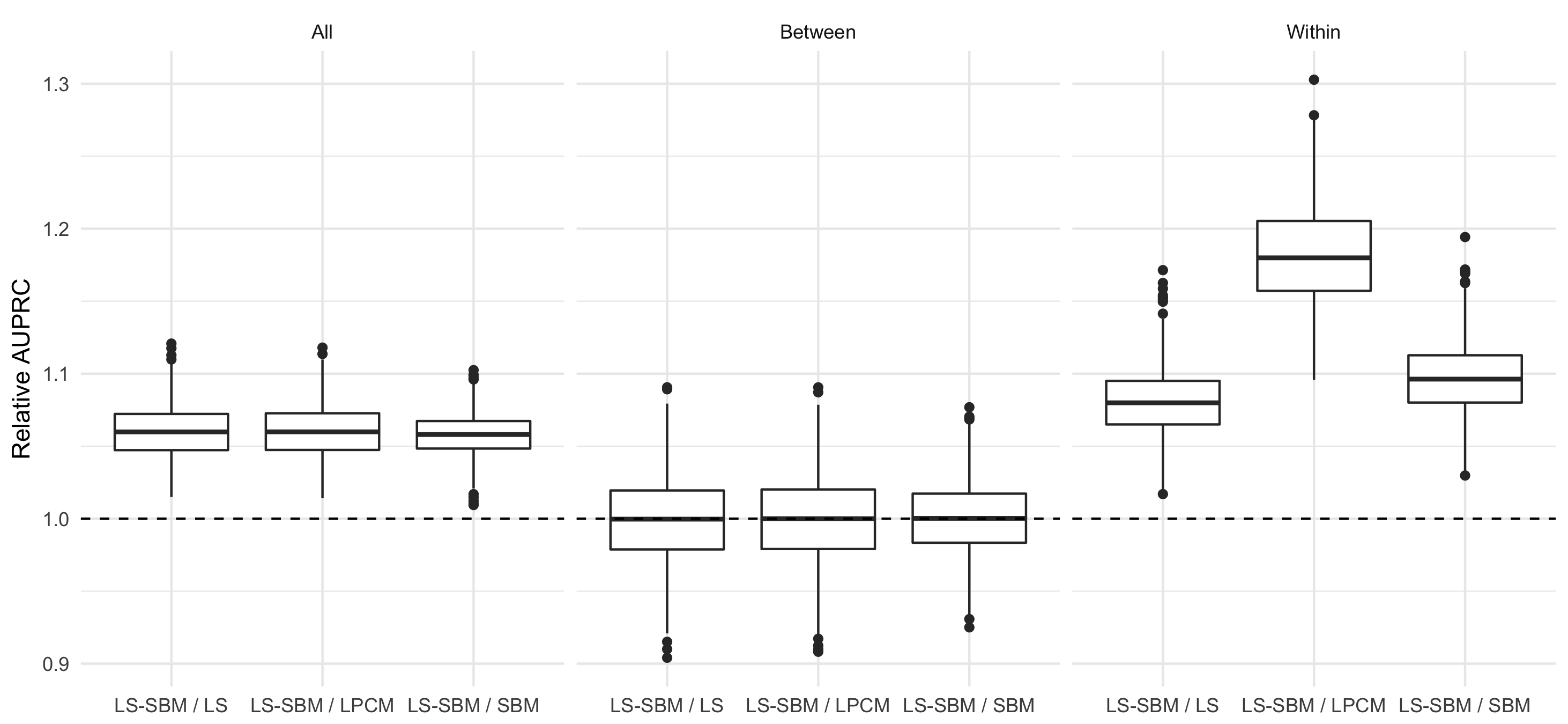}
\caption{Relative area under the precision-recall curve (AUPRC) based on out-of-sample predictions for the LS-SBM and three existing models: latent space (LS) model, latent position cluster model (LPCM) and stochastic blockmodel (SBM).  The leftmost panel shows the relative AUPRC for all held out edges, the middle panel shows the results for edges that are between nodes that are in different blocks and the right panel shows that for edges between nodes within the same block.}
\label{fig:sim}
\end{figure}
\end{landscape}

From the leftmost panel of Figure \ref{fig:sim} it is evident the LS-SBM outperforms all three existing models as the relative AUPRCs are greater than one for all simulations and all network models. Separate precision-recall curves were constructed for the held out portions of the network corresponding to relationships between nodes within the same block (rightmost panel of Figure \ref{fig:sim}) and those portions between nodes that reside in different blocks (middle panel of Figure \ref{fig:sim}). These illustrate that while the LS-SBM appears to predict edges between nodes in different blocks as well as existing models, there are notable improvements in predictions for ties between nodes within the same block.

\section{Projectivity of multiresolution network models}
\label{sec:properties}

Focusing on inference, we seek to understand which features of the hypothetical infinite population we can reasonably expect to learn from a sample of $N$ nodes. 
Projectivity is essential for inference as it facilitates comparison of model parameters across networks of different sizes. This notion of ``different sizes'' naturally arises in multiple network samples, 
which are almost certainly never be of the same size (because of the complexities of sampling networks and prevalence of missing data). In the case of the multiresolution framework, these sizes may also refer to the inferred block sizes. 

\citet{shalizi2013consistency} investigate the projectivity of families of statistical network models, where a model family $\{P_{\theta,N}: N \in \mbb{N}, \theta \in \Theta\}$ is deemed projective if distribution $P_{\theta,N}$ for a sample of $N$ actors can be recovered by marginalizing the distribution $P_{\theta,M}$, for $N < M$, over actors $\{N+1, \dotsc, M\}$.  Stated more formally, a family of network models is projective if 
$P_{\theta,N} = P_{\theta,M} \circ \pi_{M \mapsto N}^{-1}$ for all $N < M < \infty$, where $\pi_{M \mapsto N}$ is the natural projection map that selects the subgraph on the first $N$ nodes from the full graph on $M$ nodes 
and $\circ$ denotes function composition.  Letting $\mbb{Y}_{M \backslash N}$ be $\mbb{Y}_M$ after removing the $\mbb{Y}_N$ subgraph and $\s{Y}_{M\backslash N}$ be its sample space, we can write
\begin{align*}P_{\theta,N}(\mbb{Y}_N) &= P_{\theta,M}(\pi_{M \mapsto N}^{-1}(\mbb{Y}_N)) = \hspace{-.3in}\sum_{{\mbb{Y}_M \in \pi_{M\mapsto N}^{-1}(\mbb{Y}_N)}} \hspace{-.2in}P_{\theta,M}(\mbb{Y}_M) = P_{\theta, M}(\mbb{Y}_N, \mbb{Y}_{M \backslash N} \in \s{Y}_{M\backslash N}), \end{align*}
where $\pi_{M \mapsto N}^{-1}(\mbb{Y}_N)$ is the set of graphs on $\{1, \dotsc M\}$ that have $\mbb{Y}_N$ as the subgraph on the first $N$ actors.

To see why projectivity is crucial for comparisons across networks of different sizes, or equivalently blocks,   
recall the Karnataka dataset.  Suppose we have a model family and consider two village networks: village A network containing $100$ households and  
village B network containing $100,000$ households.  
Upon observing these networks, a researcher wishes to formally compare them by fitting the statistical model to each one and comparing the parameter estimates.  In order for this comparison to be meaningful, the statistical model \textit{must} be projective.  Suppose the parameter associated with the generation of network A is $\theta_A$, the parameter generating network B is $\theta_B$, and $\theta_B = \theta_A$.  The statistical model is projective if, when 90,900 households are marginalized over in the network model fit to network B, the resulting probability model on the remaining 100 households is equal to the model on network A.  (Note in this discussion, we assume the probability model is row and column exchangeable, i.e. node exchangeable as in \citet{von1991finite}.)

Our multiresolution framework proposes 
projective models, $W$ and $B$, for capturing within- and between-community relations, and combines these to form a model for the entire network.  We demonstrate below that a model class defined using combinations of projective distributions forms a projective family of models.  This permits 
researchers to make coherent comparisons 
across communities within the same network, even if communities are of different sizes.

We start by proving that mixtures of projective models are projective.  In the definition of multiresolution network models in \eqref{eq:multires}-\eqref{eq:between}, we assume that $W$ and $B$ are projective. Further, since $W_\alpha$ and $B_\omega$ are not indexed by $N$, $R_\alpha$ and $S_\omega$ must be the same regardless of the number of nodes in the graph.  Below we show that this implies that $W_\alpha$ and $B_\omega$ are projective by showing that, in general, a mixture of projective models is also projective.

\begin{restatable}{thm}{missing_data_proj}\label{thm:missing_proj}
Suppose $\{\widetilde{P}_{\theta,M}: M \in \mbb{N}, \theta \in \Theta \}$ is a projective collection of statistical models over networks and latent variables, such that $\widetilde{P}_{\theta,M}$ is a distribution on $(\mbb{Y}_M,\eta)$ supported over $\s{Y}_M \times \s{N}$ where the dimension of $\s{N}$ does not depend on $M$. Let $P_{\theta,M} = \widetilde{P}_{\theta,M} \circ \tau_M^{-1}$ for $\tau_M : \s{Y}_M \times \s{N} \to \s{Y}_M$ the projection map. Then $\{P_{\theta,M}: M \in \mbb{N}, \theta \in \Theta \}$ is a projective family as well.
\end{restatable}

\begin{proof}
Let $N < M$.  Further, let $\tilde\pi_{M \mapsto N}$ be the projection map from $\s{Y}_M \times \s{N} \to \s{Y}_N \times \s{N}$ and $\pi_{M\mapsto N}$ be the projection map from $\s{Y}_M$ to $\s{Y}_N$. Let's first suppose that \begin{equation}
    \tilde\pi_{M\mapsto N}^{-1} \circ \tau_N^{-1} = \tau_M^{-1} \circ \pi_{M\mapsto N}^{-1}. \label{eq:mapEQ}
\end{equation} Then, $ P_{\theta,N} = \widetilde{P}_{ \theta,N} \circ \tau_N^{-1} = \widetilde{P}_{\theta,M} \circ \tilde\pi_{M\mapsto N}^{-1} \circ \tau_N^{-1} = \widetilde{P}_{\theta,M} \circ  \tau_M^{-1} \circ \pi_{M\mapsto N}^{-1} = P_{\theta,M} \circ \pi_{M\mapsto N}^{-1},$
where the second equality follows from the projectivity of  $\{\widetilde{P}_{\theta,M}: M \in \mbb{N}, \theta \in \Theta \}$ and  third equality is a consequence of \eqref{eq:mapEQ}. Thus, if \eqref{eq:mapEQ} holds, $\{P_{\theta,M}: M \in \mbb{N}, \theta \in \Theta \}$ is projective by definition.

Verifying \eqref{eq:mapEQ} is straightforward. Let $\mbb{Y}_N \in \s{Y}_N$. Then
\begin{align*}
 \Big(\tilde\pi_{M\mapsto N}^{-1} \circ \tau_N^{-1}\Big)(\mbb{Y}_N) &= \tilde\pi_{M\mapsto N}^{-1}\Big(\!\big\{ (\mbb{Y}_N, \eta) \! : \! \eta \! \in \! \s{N}\big\}\!\Big) = \Big\{ \!(\mbb{Y}_N, \mbb{Y}_{M\backslash N}, \eta) \! :\! \mbb{Y}_{M\backslash N} \! \in \! \s{Y}_{M\backslash N} , \eta \! \in \! \s{N} \! \Big\}.
 \end{align*}
Similarly, \begin{align*}
\Big( \tau_M^{-1} \circ \pi_{M\mapsto N}^{-1}\Big)(\mbb{Y}_N) &= \tau_M^{-1}\Big( \! \big\{ (\mbb{Y}_N, \mbb{Y}_{M\backslash N})\! :\! \mbb{Y}_{M\backslash N}\! \in \!\s{Y}_{M\backslash N}\big\}\!\Big) \\
&= \Big\{\! (\mbb{Y}_N, \mbb{Y}_{M\backslash N}, \eta)\! :\! \mbb{Y}_{M\backslash N}\! \in \!\s{Y}_{M\backslash N} , \eta \!\in\! \s{N}\!\Big\}.
\end{align*}
\end{proof}

By a similar argument, we can show that node-level latent variable models, such as the latent space network model specified for the within-block ties in the LS-SBM, are also projective.  Using Theorem 1, we now show that the class of multiresolution models is projective.
\begin{restatable}{thm}{multires_projectivity}\label{thm:multires_proj}
Multiresolution network models are projective.
\end{restatable}
\begin{proof}
Since $W$ and $B$ are projective models, the models $W_{\alpha}(\mbb{Y}_{N, kk})$ and $B_{\omega}(\mbb{Y}_{N, kl})$ are then projective  by Theorem 1 because their distributions do not depend on $N$ and they are mixtures over projective models.  Consider 
\begin{align}
 P_{\gamma, \alpha, \omega,M}(\mbb{Y}_N,& \mbb{Y}_{M \backslash N} \in \s{Y}_{M\backslash N}) = \prod_{k=1}^{K_M} W_{\alpha}(\mbb{Y}_{N,kk}, \mbb{Y}_{M \backslash N, kk} \in \s{Y}_{M\backslash N, kk}) \notag \\
 &\qquad \prod_{k=1}^{K_M - 1} \prod_{l = k+1}^{K_M} B_{\omega}(\mbb{Y}_{N,kl}, \mbb{Y}_{M \backslash N, kl} \in \s{Y}_{M\backslash N, kl}). \label{eq:projMR}
 \end{align}
For any $k,l$ such that none of the nodes in blocks $k,l$ are in $\mbb{Y}_N$, we have $W_{\alpha}(\mbb{Y}_{N,kk}, \mbb{Y}_{M \backslash N, kk} \in \s{Y}_{M\backslash N, kk})  = W_{\alpha}(\mbb{Y}_{M \backslash N, kk} \in \s{Y}_{M\backslash N, kk}) = 1$ and similarly for $B_{\omega}$. Hence we only need consider blocks with at least one node from $\mbb{Y}_N$, and the right hand side of \eqref{eq:projMR} is equal to 
\[ \prod_{k=1}^{K_N} W_{\alpha}(\mbb{Y}_{N,kk}, \mbb{Y}_{M \backslash N, kk} \in \s{Y}_{M\backslash N, kk})  \prod_{k=1}^{K_N - 1} \prod_{l = k+1}^{K_N} B_{\omega}(\mbb{Y}_{N,kl}, \mbb{Y}_{M \backslash N, kl} \in \s{Y}_{M\backslash N, kl}). \]
Since $W_{\alpha}$ and $B_{\omega}$ are 
projective,  $ W_{\alpha}(\mbb{Y}_{N,kk}, \mbb{Y}_{M \backslash N, kk} \in \s{Y}_{M\backslash N, kk}) = W_{\alpha}(\mbb{Y}_{N,kk}),$
 and similarly for $B_{\omega}$. Thus the right hand side of \eqref{eq:projMR} is equal to
\[ \prod_{k=1}^{K_N} W_{\alpha}(\mbb{Y}_{N,kk})  \prod_{k=1}^{K_N - 1} \prod_{l = k+1}^{K_N} B_{\omega}(\mbb{Y}_{N,kl}),\]
which equals $P_{\gamma, \alpha, \omega, N}(\mbb{Y}_N)$.
\end{proof}

We emphasize that projectivity of the multiresolution model only occurs when $B_\omega$ and $W_\alpha$ are both projective, and hence these distributions do not depend on $N$.  If the between-block tie probabilities depend on the size of the observed graph, then the model is not projective overall but is projective within each block.  By ``projective within each block" we mean that the within-block parameters are still comparable across blocks and networks of different sizes, although it is meaningless to, in general, compare the between-block parameters. An advantage of the between-block parameters depending on $N$ is that the model can be sparse in the limit. That is, as the number of nodes grows from $N$ to infinity, the  expected density (i.e. proportion of edges present in the graph) goes to zero.  Equivalently, a graph is asymptotically sparse if the average expected degree grows sub-linearly with $N$.  The Aldous-Hoover Theorem implies that infinitely exchangeable sequences of nodes correspond to dense graphs~\citep{aldous1981representations,orbanz2015bayesian}.  In our case, however, for a fixed $N$ and $\gamma$, the multiresolution model is only exchangeable modulo $\gamma$, meaning that nodes within the same block are exchangeable (similar to that in regression; see  \cite{mccullagh2005exchangeability}).  If we assume each node has an equal probability of belonging to any block (e.g. placing a uniform distribution on each $\gamma_i$), the model is finitely exchangeable \citep{von1991finite}.  However, we are unaware of a prior distribution on $\gamma$ that gives both finite exchangeability and projectivity, but does not assume knowledge about a bound on the size of blocks. In the web-based supplementary material we show that families of projective models which are finitely exchangeable are asymptotically dense.  Nevertheless, for practical purposes, our model can be arbitrarily sparse for any network with a finite number of nodes and, since it is projective, permit comparison between networks of different sizes.

Our results on exchangeability and projectivity relate to recent work on nonparametric generative models for networks (e.g.~\citet{caron2014sparse},~\citet{veitch2015class},~\citet{crane2015framework}, ~\citet{crane2016edge}, or ~\citet{broderick2016edge}).  
The  objectives in our framework are subtly but critically different, however.  In recent developments in nonparametrics, the objective is to understand the structure of the graph that is implied by a probability model as the network grows from the observed size $N$ to its limit.  These recent works define new notions of exchangeability that imply critically different graph properties than  node exchangeability discussed here.  Our work, in contrast, focuses on the inverse inference paradigm: given a sample from an infinite population, our goal is to understand which properties of the population could be feasibly estimated using the multiresolution model.  This perspective leads to a focus on projectivity.

\section{Discussion}
\label{sec:discuss}
In this paper we present a multiresolution model for social network data.  Our model is well-suited for graphs that are overall very sparse, but contain pockets of local density.  Our model utilizes mixtures of projective models to separately characterize tie structure within and between dense pockets in the graph.  Our proposed framework is substantially more flexible than existing latent variable approaches (such as the LPCM) and supports meaningful comparisons of parameter estimates across communities and networks of varying sizes.

We introduced the LS-SBM as one example of a model within the multiresolution class.  However, alternative multiresolution models could be defined by replacing the latent space model representing within-community relations with LPCMs or LS-SBMs. This would add complexity to the model, allowing for sub-community structure within the global-community structure.  A key distinction between the general multiresolution framework propose here and previous approaches of, for example, Peixoto and Lyzinski et al. is that our approach allows and suggests using different network models to capture structure at different levels.  This allows models to be constructed that leverage the advantages (e.g. parsimony, detailed structure) of multiple models simultaneously.

Our work also contributes to an active discussion in the statistics literature emphasizing the importance of understanding the relationship between sampling and modeling social networks.~\citet{crane2015framework}, for example, present a general framework for sampling and inference for network models.  Our work emphasizes the importance of ``consistency under sampling'' for comparison across networks and across communities within the same graph.  

Our projectivity result for multiresolution models means that our framework can be used to compare across communities within a graph, even if the communities are different sizes.  \citet{schweinberger2014local} define the concept of domain consistency, which can be thought of as projectivity over communities, and propose a class of models that satisfy this property. If, for example, a member of a community is missing, then the interpretation of the parameters describing behavior in that community will be fundamentally different than if the member were present.  Our model also has this property, but is more general.  In particular, the notion of projectivity over communities requires that data be collected by using cluster sampling over communities. 
Here we strengthened and generalized this framework by proving that our model is  projective at the actor level, i.e.\ even if complete communities are not observed.  Using our framework, it is possible to sample at the actor rather than the community level while estimating parameters that are comparable across communities.

\paragraph{Acknowledgements}
This work was partially supported by the National Science Foundation under Grant Number SES-1461495 to Fosdick and Grant Number SES-1559778 to McCormick.  McCormick is also supported by grant number K01 HD078452 from the National Institute of Child Health and Human Development (NICHD).  This material is based upon work supported by, or in part by, the U. S. Army Research
Laboratory and the U. S. Army Research Office under contract/grant number W911NF-12-1-0379.  Murphy and Ng are supported by the Science Foundation Ireland funded Insight Research Centre (SFI/12/RC/2289).  The authors would also like to thank the Isaac Newton Institute Program on Theoretical Foundations for Statistical Network Analysis workshop on Bayesian Models for Networks, supported by EPSRC grant number EP/K032208/1.

\bibliographystyle{abbrevnamed.bst}

\bibliography{multiresbib}

\clearpage
\centerline{\Huge \bf Web-based supplementary materials}

\appendix

\doublespace

\section{Interplay between finite exchangeability, projectivity, and asymptotic sparsity}
The relationship between (finite) exchangeability and projectivity is an important part of our modeling framework.  Since we restrict the size of each community to be finite, our model is finitely exchangeable (and not infinitely exchangeable).  We now show that, even with finite exchangeabiltiy, a model that is projective is also dense in the limit. 

\begin{restatable}{thm}{proj_density}\label{thm:proj_density}
Projective families of finitely node-exchangeable models for networks are asymptotically dense.
\end{restatable}
\begin{proof}
First we write
\[ E_{P_{N,\theta}}\left[ \sum_{\mbb{Y}_N} Y_{ij} \right] = \sum_{\mbb{Y}_N} P_{N,\theta}( Y_{ij} = 1). \]
Now we can write the marginal probability as
\[ P_{N, \theta}(Y_{ij} = 1) = P_{N, \theta}(Y_{12}) = P_{N, \theta}(Y_{12} = 1, \mbb{Y}_{N \backslash 12} \in \s{Y}_{N \backslash 12}) =  P_{2,\theta}( Y_{12} )  \]
where the first equality was by exchangeability of the distribution $P_{N, \theta}$ and the second two by projectivity. Hence the probabilities in the sum are all the same and the expected number of edges in the newtork is 
\[ E_{P_{N,\theta}}\left[ \sum_{\mbb{Y}_N} Y_{ij} \right] = \tfrac{1}{2}N(N-1) P_{2,\theta}( Y_{12} ).\]
Hence the family is asymptotically dense with asymptotic density $P_{2,\theta}( Y_{12} )$.
\end{proof}

\section{Prior specification in the LS-SBM}
The assortativity restriction discussed in Section \ref{sec:priors} and shown in \eqref{eq:assort} required the mean (logit) probability of a tie within blocks be greater than or equal to the mean (logit) probability of a tie between blocks.  To translate this restriction into constraints on the parameters, first consider the within-block tie probabilities. It can be shown $$E\Big[\|\mb{Z}_i - \mb{Z}_j\|\Big] = 2 \sigma_k \frac{\Gamma(\frac{D+1}{2})}{\Gamma(\frac{D}{2})}.$$ Then, if $\mb{Z}_1,...,\mb{Z}_N$ are IID N$_D$($0,\sigma_k^2 I_D$), where $D$ is the dimension of the latent space, the conditional expected logit is equal to $$E\Big[\mathrm{logit}(Pr(Y_{ij}=1))|\gamma_i = \gamma_j = k,\beta_k,\sigma_k\Big] = \beta_k -  2 \sigma_k \frac{\Gamma(\frac{D+1}{2})}{\Gamma(\frac{D}{2})}.$$  Taking expectations then over the distribution of $\eta_k = (\beta_k,\sigma_k)$, we find
\begin{equation}
E\Big[\mathrm{logit}(Pr(Y_{ij}=1))|\gamma_i = \gamma_j\Big] = \mu_1 -  2 e^{\mu_2} \frac{\Gamma(\frac{D+1}{2})}{\Gamma(\frac{D}{2})}.\label{eq:assort_within}
\end{equation}
Recall the between block tie probabilities $Pr(Y_{ij} = 1|\gamma_i \not= \gamma_j, \tau_{\gamma_i\gamma_j}) = \tau_{\gamma_i\gamma_j}.$  Then for $\gamma_i \not=\gamma_j$, $\mathrm{logit}(Pr(Y_{ij}=1)) = \log\Big(\frac{\tau_{\gamma_i\gamma_j}}{1-\tau_{\gamma_i\gamma_j}}\Big)$ and marginalizing over the Beta distribution, we find
\begin{equation}
E\Big[\mathrm{logit}(Pr(Y_{ij}=1))|\gamma_i \not= \gamma_j\Big] = \psi(a_0) - \psi(b_0), \label{eq:assort_between}
\end{equation}
where $\psi(x)$ is the digamma function: $\psi(x) = \frac{d \text{log}(\Gamma(x))}{dx}$. 

Combining equations \eqref{eq:assort_within} and \eqref{eq:assort_between}, the restriction on the population parameter space is
$$\mu_1 -  2 e^{\mu_2} \frac{\Gamma(\frac{D+1}{2})}{\Gamma(\frac{D}{2})} \ge  \psi(a_0) - \psi(b_0).$$  
This constraint is incorporated in the prior specification in \eqref{eq:prior} in the manuscript.

\section{Steps for choosing number of blocks}

We now detail our method of pre-selecting the number of blocks $K$ for a network using cross-validated spectral clustering. Let $K_{\mathrm{max}}$ be the maximal number of clusters the researcher wants to consider.

Let $\mathcal{A} = \{(i, j) : 1 \leq i < j \leq N\}$ be the set of possible edges in an undirected binary network with no self-loops.

\begin{enumerate}
    \item Create a square matrix $P_k$ of the same dimensions as the observed network $Y$ for each $k = 1, \dotsc, K_{\mathrm{max}}$. $P$ will contain the held-out predicted probabilities for  probabilties.
    \item Randomly partition $\mathcal{A}$ in to $M_{CV}$ folds of equal size $\mathcal{A}_1, \dotsc, \mathcal{A}_{M_{CV}}$ (or as close to equal size as possible).
    \item For $m=1, \dotsc, M_{CV}$:
    \begin{enumerate}
        \item Create a network $\tilde{Y}^{(m)}$ which is a copy of the observed network $Y$ except that $\tilde{Y}_{i,j}^{(m)}$ is missing for all $(i,j) \in \mathcal{A}_{m}$.
        \item Impute the elements of $\mathcal{A}_{m}$ in $Y^{(m)}$ using the observed degrees of each node (i.e.\ $\tilde{Y}_{i,j}^{(m)} = d_i d_j$ where $d_i$ is the fraction of observed edges among non-missing possibilities for node $i$).
        \item For each $k=1, \dotsc, K_{\mathrm{max}}$:
        \begin{enumerate}
            \item Estimate a degree-corrected stochastic block model  on $\tilde{Y}^{(m)}$ with $k$ blocks using, e.g.\  assortative spectral clustering.
            \item Re-impute the elements of $\mathcal{A}_{m}$ in $Y^{(m)}$ using the predicted probabilities from the estimated degree corrected stochastic block model.
            \item Repeat until the sum of the squared differences in predicted probabilities from one iteration to the next falls below a desired threshold.
            \item Set the elements of $\mathcal{A}_m$ in $P_k$ to be the predicted probabilities from the final stochastic block model fit.
        \end{enumerate}     
    \end{enumerate}
    \item Calculate the AUC, MSE, and MPI between $P_k$ and $Y$ for each $k$.
\end{enumerate}
We repeated the above procedure twenty times (with different random folds each time) to obtain ten out-of-sample AUC, MSE, and MPI estimates for each possible $K$. We then computed the mean AUC, MSE, and MPI and a 95\% confidence interval for this mean for each $K$. Each measure had a value of $K$ that minimized the measure -- for each measure we defined the selected number of blocks based on that measure as the smallest $K$ whose 95\% confidence interval contained the mean of the optimal $K$.

\section{LPCM on Karnataka village dataset}

As a comparison, we also present the fit from the LPCM in Figure~\ref{fig:lpcm}.  We fit the LPCM using the variational approximation of~\citet{salter2013variational} provided in the $\mathtt{R}$ package `VBLPCM'~\citep{salter2015package}.  We used six clusters in the LPCM for comparison with our six blocks.  A first key distinction is that the LPCM has a smaller global intercept term than $\mu_1$ and is forced to capture heterogeneity in tie propensity by expanding the distance between clusters (and thereby individuals) in the latent space.  This approach is in contrast to the LS-SBM which maintains block specific intercepts and latent spaces, facilitating comparison across blocks.  Consequences of encoding all clusters in the same latent space is that tie propensities are extremely small between groups on opposite sides of the latent space and relations between groups that are adjacent are constrained by the triangle inequality.  The group represented in green, for example, is adjacent to the group in pink, but by virtue of the distance between the pink and teal groups, must also be close to the group in teal.

\begin{figure}[ht!]
\centering
\includegraphics[width=3.5in]{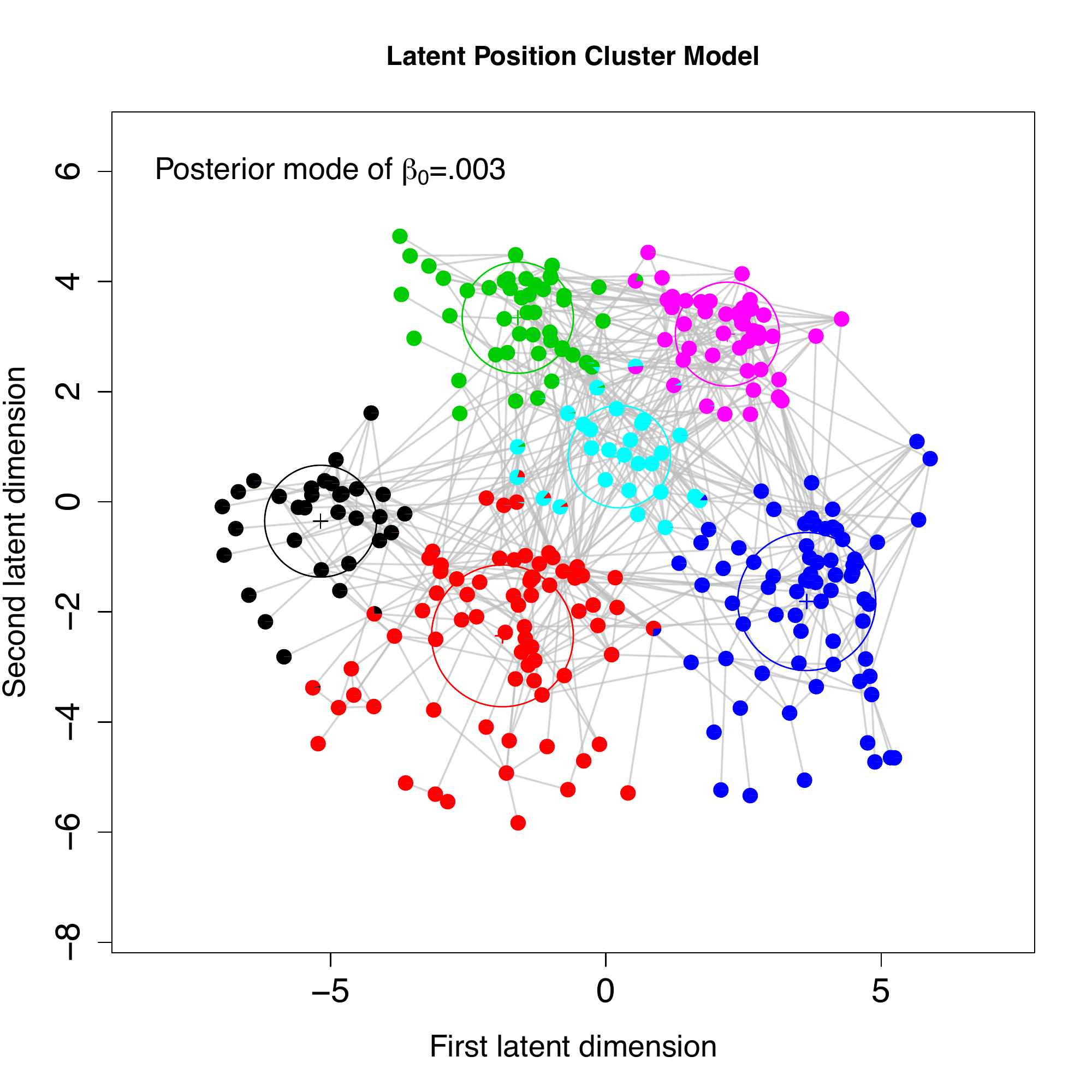}
\caption{LPCM latent positions for the household-level ``visit" relation data for village number 59 from the Karnataka village dataset.}
\label{fig:lpcm}
\end{figure}

\section{Sampling algorithm}
In this section, we give the details of the sampling algorithm. For priors we set
\begin{align*}
\gamma_i &\sim \text{Categorical}(\boldsymbol{\pi}) \hspace{.3in} i=1,...,N \\
\boldsymbol{\pi} &\sim \text{Dirichlet}(\upsilon_0, \dotsc, \upsilon_0)\\
\bs{\tau}_{kl} &\sim \text{Beta}(a_{0}, b_{0}) \hspace{.3in} 1 \le k < l \le N\\
(\boldsymbol{\mu} | \boldsymbol{\Sigma}) &\sim \text{MVN}(\mathbf{m}_0, s_0^{-1} \boldsymbol{\Sigma}) \\
\boldsymbol{\Sigma} &\sim \text{InvWishart}(\boldsymbol{\Psi}_0, \nu_0).
\end{align*}
Denote the full set of parameters $\boldsymbol{\zeta} =\{ \boldsymbol{\gamma, \bs{\tau},\mb{Z},\beta, \pi, \sigma,\mu, \Sigma}\}$. Also define $\alpha_k = (\beta_k, \log\sigma_k)$ as the within-block latent space parameters. The posterior factors as
\begin{align*}
P(\boldsymbol{\zeta}|Y)\propto P(Y|\boldsymbol{\gamma,\bs{\tau},\mb{Z},\beta})& P(\boldsymbol{\gamma} | \boldsymbol{\pi})P(\mb{Z} | \boldsymbol{\sigma}) P(\boldsymbol{\beta}, \boldsymbol{\sigma}, \bs{\tau} | \boldsymbol{\mu}, \boldsymbol{\Sigma}, a_0, b_0) \\
  &\qquad \times P(\boldsymbol{\pi} | \upsilon_0)P(\boldsymbol{\mu} | \boldsymbol{\Sigma}, \mathbf{m}_0, s_0) P( \boldsymbol{\Sigma} | \boldsymbol{\Psi}_0, \nu_0),
  \end{align*}

The full posterior is not available in closed form.  We thus take draws from the posterior using the Markov chain Monte Carlo algorithm below (all parameters besides the one being updated are understood be set to their latest values).  

Let $n_k = \sum_i \mathbf{1}_{\gamma_i= k}$ denote the number of nodes in block $k$, $\bar{\bs{\alpha}}^{(t)} = \frac{1}{K} \sum_k \bs{\alpha}_k^{(t)}$ be the sample mean of the block-level $\alpha$ parameters, and $\mb{S}_{\alpha}^{(t)} = \sum_k \big( \bs{\alpha}_{k}^{(t)} - \bar{\bs{\alpha}}^{(t)}\big)\big( \bs{\alpha}_{k}^{(t)} - \bar{\bs{\alpha}}^{(t)}\big)^T$.  Given an admissible set of initialization values,  iteration $t+1$ of the sampling algorithm proceeds as follows:

\begin{enumerate}
\item For $i = 1, 2, \dotsc, N$: 
\begin{enumerate}
\item Propose $\gamma_{i}^* \sim \text{Categorical}(\lambda_{i1}^{(t)}/ \Sigma_k \lambda_{ik}^{(t)},....,\lambda_{iK}^{(t)}/ \Sigma_k \lambda_{ik}^{(t)})$ where
\begin{equation*} \lambda_{ik} =  \frac{\epsilon + \sum_{j \in \s{S}_k} Y_{ij}}{|\s{S}_k | + 1} \end{equation*}
In words, the probability that the proposal for node $i$ is block $k$ is proportional to the number of ties $i$ has to the block plus $\epsilon$ divided by the size of the block plus one. The presence of $\epsilon$ avoids probabilities equal to 0 or 1 and encourages jumping to blocks with few nodes. 

\item Conditional on this configuration of group memberships, propose $\mb{Z}_i^* | \gamma_i^* \sim \text{MVN}_D(\mb{m}_i^*, r_Z^2 I_D)$, where
\begin{equation*} \mb{m}_i^* = \begin{cases} \mb{Z}_i^{(t)}, & \gamma_i^* = \gamma_i^{(t)} \\
									     \tfrac{1}{|\s{G}_{i, \gamma_i^*}^{(t)} |}\sum_{j \in \s{G}_{i, \gamma_i^*}^{(t)}} \mb{Z}^{(t)}_j, & \gamma_i^* \neq \gamma_i^{(t)} \text{ and } |\s{G}_{i, \gamma_i^*}^{(t)}| > 0 \\
									     0, & \text{ otherwise}.
					   	  \end{cases} \end{equation*}
Here $\s{G}_{i, \gamma_i^*}^{(t)} = \{j : \gamma_j^{(t)} = \gamma_i^*, Y_{ij} = 1\}$ is the set of nodes in the same proposed block as $i$ to which $i$ is connected. In words, if $i$ stays in the same block, center at its last position. If it moves to a new block and has ties in that block, center at the mean position of its ties in that block. If it moves to a new block and does not have any ties in that block, center at the origin.  The variance of the proposed position coordinates equals $r_Z^2$.

\item Set $(\gamma_i^{(t+1)}, \mb{Z}_i^{(t+1)}) = (\gamma_i^*, \mb{Z}_i^*)$ with probability
\[\frac{p(\gamma_i^*,  \mb{Z}_i^* | \n{others})q(\gamma_i^{(t)}, \mb{Z}_i^{(t)} | \gamma_i^*, \mb{Z}_i^{*})   }{p( \gamma_i^{(t)}, \mb{Z}_i^{(t)} | \n{others})q(\gamma_i^*, \mb{Z}_i^{*} |\gamma_i^{(t)}, \mb{Z}_i^{(t)} ) } \]
where $q((1) | (2))$ is shorthand for the transition density to (1) from (2).  For the first ratio we have
\begin{align*}
    \frac{p( \gamma_i^*,  \mb{Z}_i^* | \n{others})}{p( \gamma_i^{(t)},  \mb{Z}_i^{(t)} | \n{others})} = &\frac{p(Y | \bs{\gamma}^*,  \mb{Z}^*, \bs{\tau}^{(t)}, \bs{\beta}^{(t)}) p(\gamma_i^* | \bs{\pi}^{(t)}) p(\mb{Z}_i^* | \bs{\sigma}^{(t)})}{p(Y | \bs{\gamma}^{(t)},  \mb{Z}^{(t)}, \bs{\tau}^{(t)}, \bs{\beta}^{(t)})p(\gamma_i^{(t)} | \bs{\pi}^{(t)}) p(\mb{Z}_i^{(t)} | \bs{\sigma}^{(t)})}.
\end{align*} 
For the second ratio, we have
\begin{align*} q({\gamma}_i^*, \mb{Z}_i^{*} |{\gamma}_i^{(t)}, \mb{Z}_i^{(t)} ) &=q({\gamma}_i^* |{\gamma}_i^{(t)} ) q( \mb{Z}_i^{*} |{\gamma}_i^*, {\gamma}_i^{(t)}, \mb{Z}_i^{(t)} ) \\
&= \left[ \prod_j (\lambda_{ij}^{(t)}/ \Sigma_k \lambda_{ik}^{(t)})^{\mathbf{1}\{\gamma_{i}^*=j\}} \right] \phi_D\left(\mb{Z}_i^*; \mb{m}_i^*(\gamma_i^*, \bs{\gamma}^{(t)}, \mb{Z}^{(t)}), r_Z^2 I_D\right) 
\end{align*}
and analogously for the numerator. 
\end{enumerate}

\item For each $i=1, \dotsc, N$ propose $\mb{Z}_i^* \sim \text{MVN}_D(\mb{Z}_i^{(t)}, r_Z^2 I_D)$, and set $\mb{Z}_i^{(t+1)} = \mb{Z}_i^*$ with probability
\[ \frac{\left[\prod_{j \in \s{S}_k}p(Y_{ij} | \mb{Z}_i^*,\mb{Z}_j, \bs{\beta}_k)\right] \phi_D(\mb{Z}_i^*; 0, \sigma_k^{2} I_D)}{\left[\prod_{j \in \s{S}_k}p(Y_{ij} | \mb{Z}_i^{(t)},\mb{Z}_j, \bs{\beta}_k)\right]\phi_D(\mb{Z}_i^{(t)}; 0,\sigma_k^{2} I_D)}\]
where here $k = \gamma_i$. Otherwise $\mb{Z}_i^{(t+1)} = \mb{Z}_i^{(t)}$. This step updates the latent positions of each node, without updating the block memberships.  

\item Sample $\pi^{(t+1)} \sim \text{Dirichlet}(\upsilon_0 + n_1^{(t)}, \dotsc, \upsilon_0 + n_K^{(t)})$.
\item For $k = 2, \dotsc, K$ and $l = 1, \dotsc, k-1$,  sample $\bs{\tau}_{kl}^{*} \sim \text{Beta}(a_{0kl} +s_{kl}, b_{0kl} + n_k n_l - s_{kl}),$ where  $s_{kl} = \sum_{i,j} Y_{ij}\mb{1}_{\gamma_i = k, \gamma_j = l}$ is the number of edges between blocks $k$ and $l$.  Set $\bs{\tau}_{kl}^{(t+1)} = \bs{\tau}_{kl}^{*}$.
\item For $k=1, \dotsc, K$, let $\bs{\alpha}_k = (\beta_k, \log(\sigma_k))$ and sample $\bs{\alpha}_k^* \sim \n{MVN}_2(\bs{\alpha}_k^{(t)}, A_{\alpha})$. Set $\bs{\alpha}_k^{(t+1)} =\bs{\alpha}_k^*$ with probability
\begin{align*}
&\frac{ \left[\prod_{i, j \in \s{S}_k}P(Y_{ij} | \mb{Z}_i,\mb{Z}_j, \beta_k^*) \right]\prod_{i \in \s{S}_k} \left[ \phi_D(\mb{Z}_i; 0,\sigma_k^{2*} I_D)\right] }{ \left[\prod_{i, j \in \s{S}_k}P(Y_{ij} | \mb{Z}_i,\mb{Z}_j, \beta_k^{(t)}) \right] \prod_{i \in \s{S}_k} \left[ \phi_D(\mb{Z}_i; 0,\sigma_k^{2(t)} I_D)\right] }\\
&\qquad\frac{\phi_D(\bs{\alpha}_k^*; \bs{\mu},\bs{\Sigma})}{\phi_D(\bs{\alpha}_k^{(t)}; \bs{\mu} ,\bs{\Sigma})}.
\end{align*}
\item Sample each component of $\boldsymbol{\mu} = (\mu_1, \mu_2)$ one at a time from their respective full conditional distribution given all other parameters and subject to the assortativity restriction.  More details given below.
\item Sample
\[\hspace{-.5in}\bs{\Sigma}^{(t+1)} \sim \text{InvWishart}\left(\bs{\Psi}_0 + \mb{S}_{\alpha}^{(t)} + \tfrac{Ks_0}{K + s_0}\big(\bar{\bs{\alpha}}^{(t)} - \mathbf{m}_0\big)\big(\bar{\bs{\alpha}}^{(t)} - \mathbf{m}_0\big)^T,  K + \nu_0\right).\]

\end{enumerate}
The above algorithm is repeated numerous times until a suitable sample from the posterior distribution is obtained.

We now describe the process of sampling from the vector of global within-community means, $\bs{\mu}$.  First note that without restricting assortativity, the (joint) full conditional is:
\[\bs{\mu}^{(t+1)} \sim \text{MVN}\left(\widetilde{\bs{m}} = \frac{ K\bar{\bs{\alpha}}^{(t)}+ s_0 \mathbf{m}_0}{K + s_0}, \widetilde{\bs{\Sigma}} =  (K + s_0)^{-1} \bs{\Sigma}^{(t)}\right).\]  Denote the components of $\widetilde{\bs{m}} = (\widetilde{m}_1, \widetilde{m}_2)$ and the components of $\widetilde{\bs{\Sigma}}$ as $\widetilde{\sigma}_1^2$, $\widetilde{\sigma}_2^2$, and $\widetilde{\rho}$.  Recall that our restriction for assortativity is:
$$E\Big[\mathrm{logit}(Pr(Y_{ij}=1))|\gamma_i = \gamma_j\Big] \ge E\Big[\mathrm{logit}(Pr(Y_{ij}=1))|\gamma_i \not= \gamma_j\Big]. $$
Given this restriction, the full conditional for $\bs{\mu}$ becomes truncated multivariate normal.  While efficient algorithms exist for sampling from truncated multivariate normals subject to linear constraints (e.g.~\citet{rodriguez2004efficient} and others), the logit function means that our constraints on the parameters space are nonlinear.  Thus, instead of sampling from the multivariate normal, we update each component of $\bs{\mu}$ conditional on the others.  

The assortativity constraint can be expressed in terms of $\bs{\mu}$ as
$$\mu_1 -  2 e^{\mu_2} \frac{\Gamma(\frac{D+1}{2})}{\Gamma(\frac{D}{2})} \ge  \psi(a_0) - \psi(b_0).$$  Updating $\bs{\mu}$ thus involves two steps:
\begin{enumerate}
\item Update $\mu_1^{(t+1)}$ subject to $$\mu_1^{(t+1)}\geq f_1(a_0,b_0,\mu_2^{(t)}) \equiv \psi(a_0) - \psi(b_0) -  2 e^{\mu_2^{(t)}} \frac{\Gamma(\frac{D+1}{2})}{\Gamma(\frac{D}{2})}$$ by drawing from the truncated normal \[\mu_1^{(t+1)}\sim N\left( \widetilde{m}_1+\frac{\widetilde{\sigma}_1}{\widetilde{\sigma}_2}\widetilde{\rho} \big(\mu_2^{(t)}-\widetilde{m}_2\big),\big(1-\widetilde{\rho}^2\big)\widetilde{\sigma}_1^{2}\right) {\bf 1}\left(\mu_1 \ge f_1(a_0,b_0,\mu_2^{(t)}) \right)\]
\item Update $\mu_2^{(t+1)}$ subject to $$\mu_2^{(t+1)} \leq f_2(a_0,b_0,\mu_1^{(t)}) \equiv \log \left( -\frac{\Gamma(\frac{D}{2})}{2\Gamma(\frac{D+1}{2})} \left(\psi(a_0) - \psi(b_0)-\mu_1^{(t)}\right) \right)$$ 
by drawing from the truncated normal \[\mu_2^{(t+1)}\sim N\left(\widetilde{m}_2+\frac{\widetilde{\sigma}_2}{\widetilde{\sigma}_1}\widetilde{\rho} \big(\mu_1^{(t)}-\widetilde{m}_1\big),\big(1-\widetilde{\rho}^2\big)\widetilde{\sigma}_2^{2}\right) {\bf 1}\left(\mu_2 \le f_2(a_0,b_0,\mu_1^{(t)}) \right).\]
\end{enumerate}
We obtained truncated normal draws using the \emph{rtnorm} function in the \emph{msm} package~(\citet{ jackson2011multi}).

\section{Sampler post-processing}

The likelihood is invariant to permutations of the block memberships and to rotations and reflections of the latent spaces. We post-process the MCMC samples to resolve these non-identifiabilities. 

The invariance of the likelihood to permutations of the block memberships is called label-switching and it has been studied extensively \citep[eg.][]{jasra2005labelswitching,rodriguez2014labelswitching}. In order to handle the label-switching invariance of the block memberships, we first fix a membership vector $\bs{\gamma}_0 = (\gamma^0_{1},....\gamma^0_{N})$ toward which to permute the memberships. While multiple choices for $\bs{\gamma}_0$ are possible, we use a random membership vector from one of the chains. For each membership sample $\bs{\gamma}_s = (\gamma_1^{(s)},...,\gamma_N^{(s)})$ we use the Hungarian algorithm to find a permutation $\omega$ of $1, \dotsc, K$ that maximizes the number of matching memberships between $\bs{\gamma}_0$ and $\bs{\gamma}_s$ relabelled according to $\omega$: $\sum_k \mathbf{1}\{\gamma^{(s)}_{\omega(k)} = \gamma^0_{k}\}$  (\citet{ papadimitriou1998combinatorial}). When there are multiple permutations maximizing this criterion we choose among them at random. Define $\bs{\tilde{\gamma}}_s$ to be $\bs{\gamma}_s$ relabelled according to the optimal permutation. 

To address the non-identifiability of the latent spaces, we adapt the ``Procrustes" transformation method used by \citet{HRH02}. For each block $k$, we first identify the set of all nodes $S_{k0}$ that have non-zero posterior probability of membership in $k$. We then construct a distance matrix $D_{k0}$ and a weight matrix $W_{k0}$ as follows. For every pair of nodes $i \neq j \in S_k$, let $(W_{k0})_{ij}$ be the number of posterior samples in which $i$ and $j$ were both members of block $k$. If $(W_{k0})_{ij} > 0$ then let $(D_{k0})_{ij}$ be the average distance between the positions of $i$ and $j$ for those samples in which they were both members of block $k$. Otherwise $(D_{k0})_{ij}$ is missing.

We then perform a weighted multidimensional scaling (MDS) of $D_{k0}$ with weights $W_{k0}$ to find the best $D$-dimensional representation of the average distances using SMACOF (\citet{ de2009multidimensional}). Call this representation $\mb{{Z}}_{k0}$. For each posterior sample $s$, call the members of block $k$ at this sample $S_{ks}$ and their positions $\mb{Z}_{ks}$. Let $\mb{\tilde{Z}}_{k0s}$ be the subset of $\mb{Z}_{k0}$ corresponding to $S_{ks}$. Compute the Procrustes transformation $\mb{Z}_{sk}$ toward $\mb{\tilde{Z}}_{k0s}$ the distance-preserving transformation of $\mb{Z}_{ks}$ (i.e.\ the transformation over all rotations, reflections, and shifts of $\mb{Z}_{sk}$ minimizing the sum of square distances between $\mb{Z}_{sk}$ and  $\mb{\tilde{Z}}_{k0s}$). 

\begin{figure}
    \centering
    \includegraphics[width=.75\textwidth]{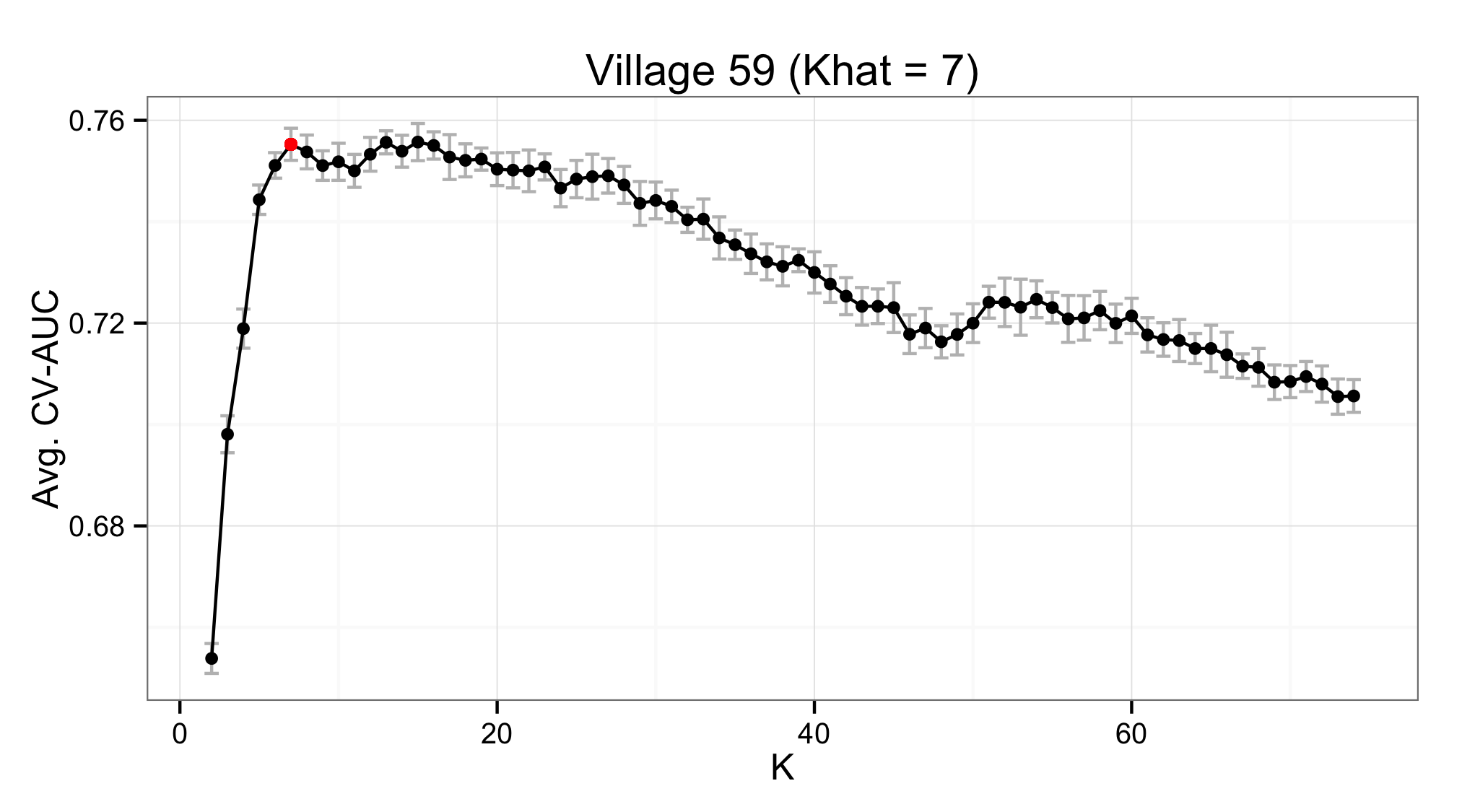}
        \includegraphics[width=.75\textwidth]{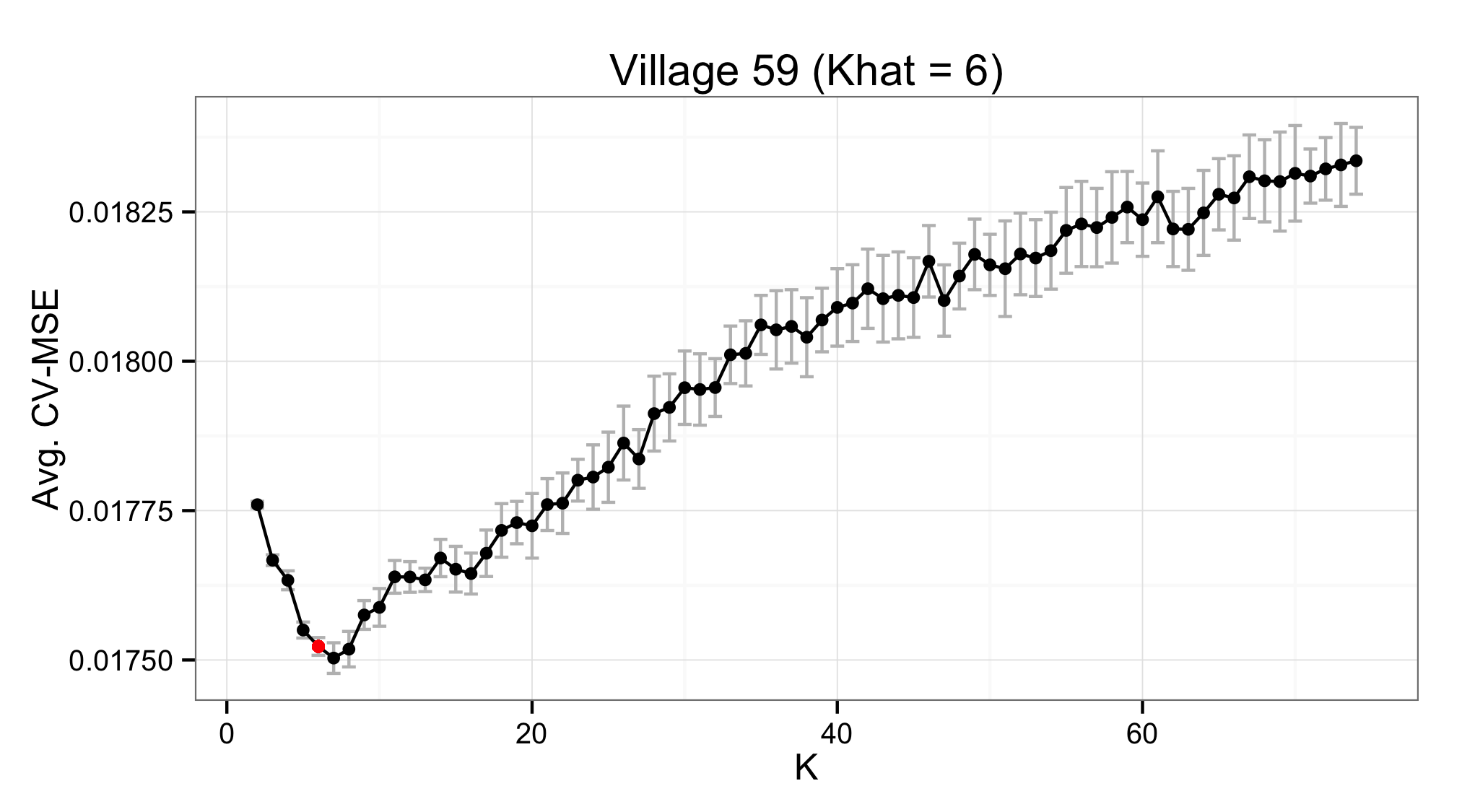}
    \includegraphics[width=.75\textwidth]{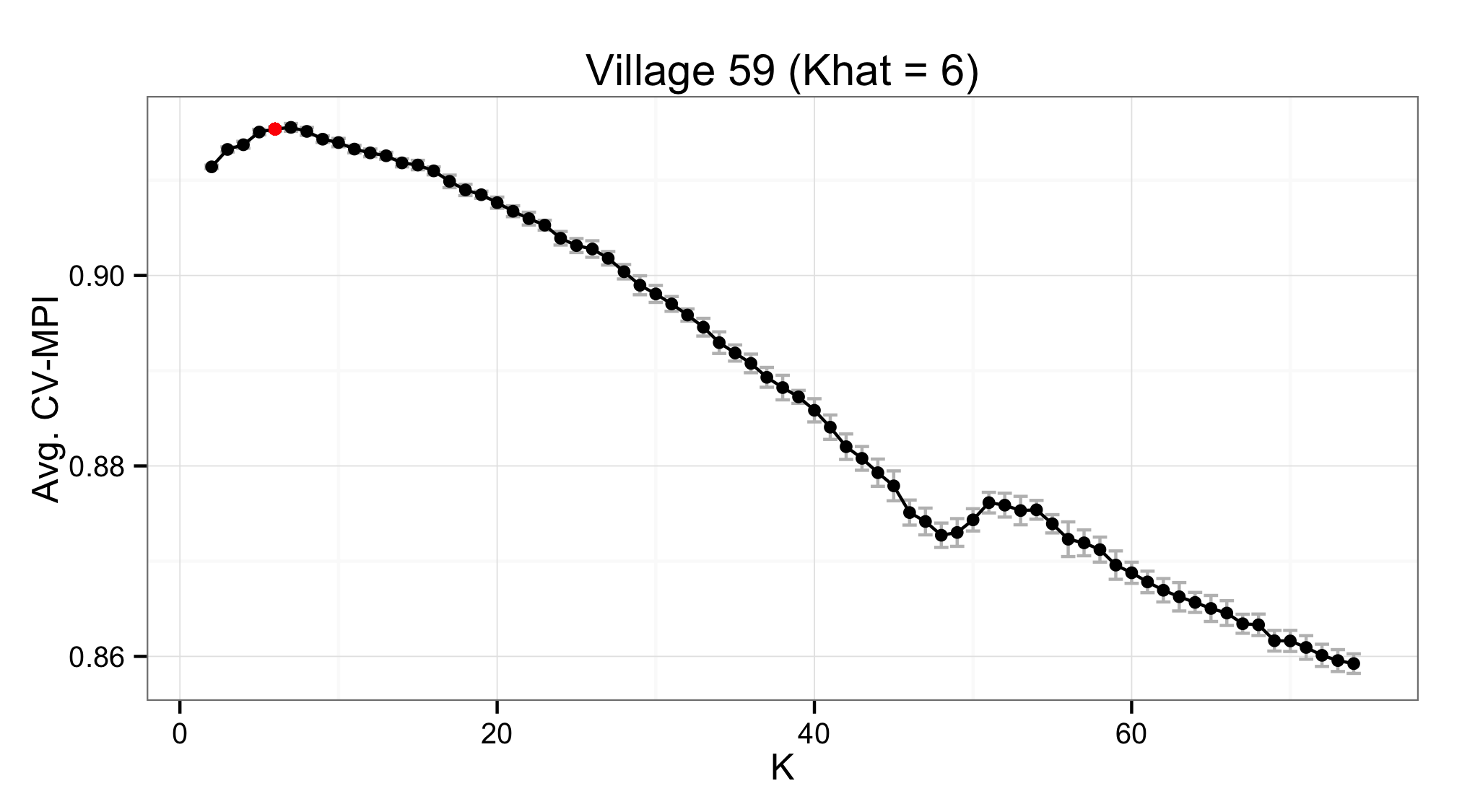}

    \caption{Diagnostic plot for choosing the number of clusters.}
\end{figure}

\section{Estimation of LS-SBM on Karnataka data}

We used the cross-validation scheme outlined in Section \ref{sec:fitting} to choose the number of blocks, partitioning the data 20 times to perform 10-fold cross validation. For both MSE and MPI, $K=6$ was the smallest $K$ value for which the estimated mean criteria was within the 95\% CI for the estimated mean criteria of the optimal value, while for AUC $K=7$ was the smallest such $K$. We chose $K=6$.

We ran four MCMC chains, each with 160,000 iterations. We kept every twentieth iteration and discarded the first quarter of each chain as burn-in.  We present traceplots and convergence diagnostics in the supplementary material.  We  fixed the hyperpriors on the between-block beta distribution to enforce global assortativity, as described in Section~\ref{sec:priors}.  Specifically, we fixed the $b_0$ parameter to be one and then chose $a_0$ such that the mean of the beta prior distribution is ten times the observed network density.  This choice produces a prior that has substantial mass between zero and about 0.4.  We also experimented with different choices for $a_0$ and $b_0$ and found that, while enforcing assortativity does have the desired impact, the substantive conclusions were similar for a wide range of choices for $a_0$ and $b_0$. Codes to replicate the results are available at \url{https://github.com/tedwestling/multiresolution_networks.git}. 

 We also evaluated the convergence across the four chains using the Gelman-Rubin statistic~(\citet{ gelman2014bayesian}) from the \emph{rstan} package~(\citet{ carpenter2016stan}). 
\begin{figure}
    \centering
    \includegraphics[width=.75\textwidth]{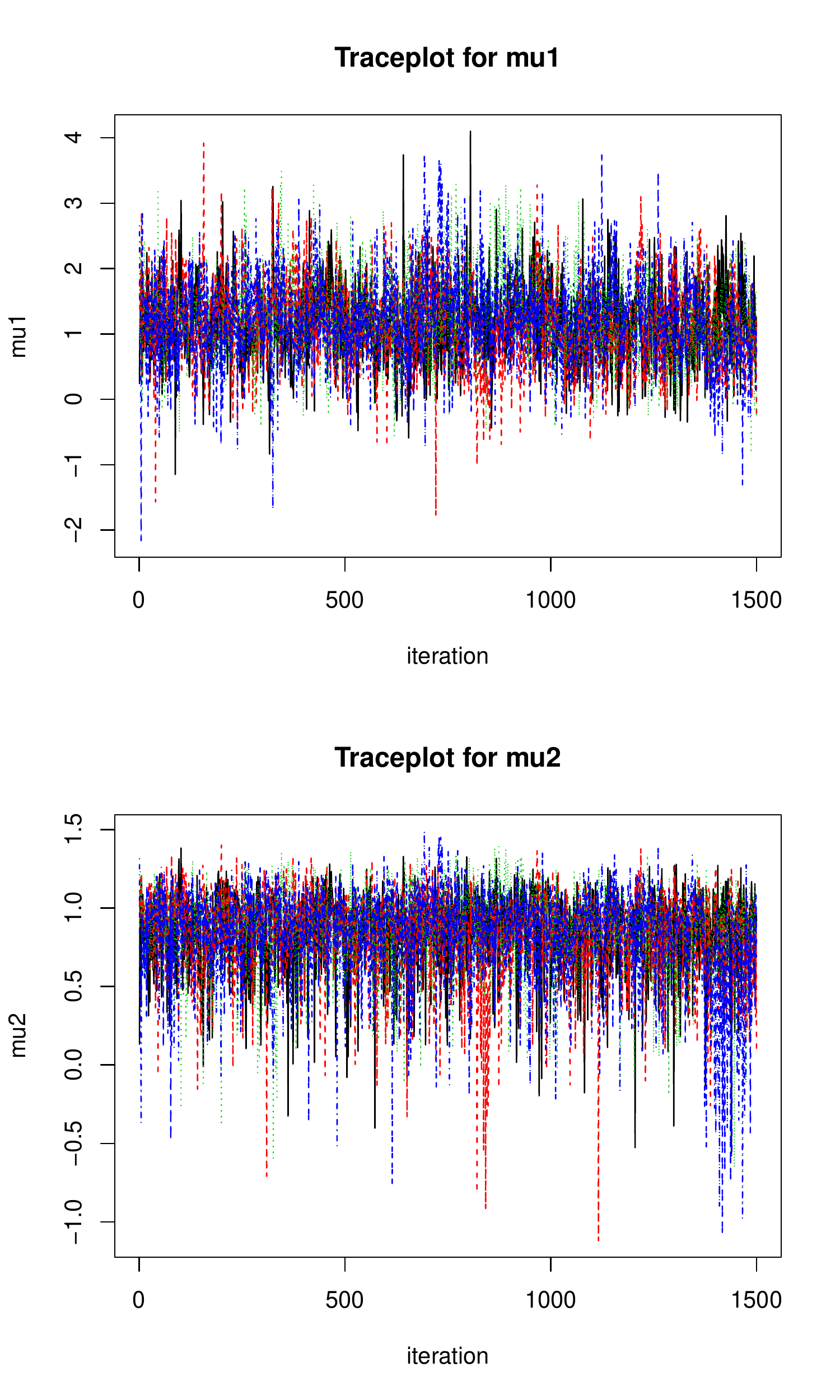}
    \caption{$\boldsymbol{\mu}$ trace}
\end{figure}

\begin{figure}
    \centering
    \includegraphics[width=.75\textwidth]{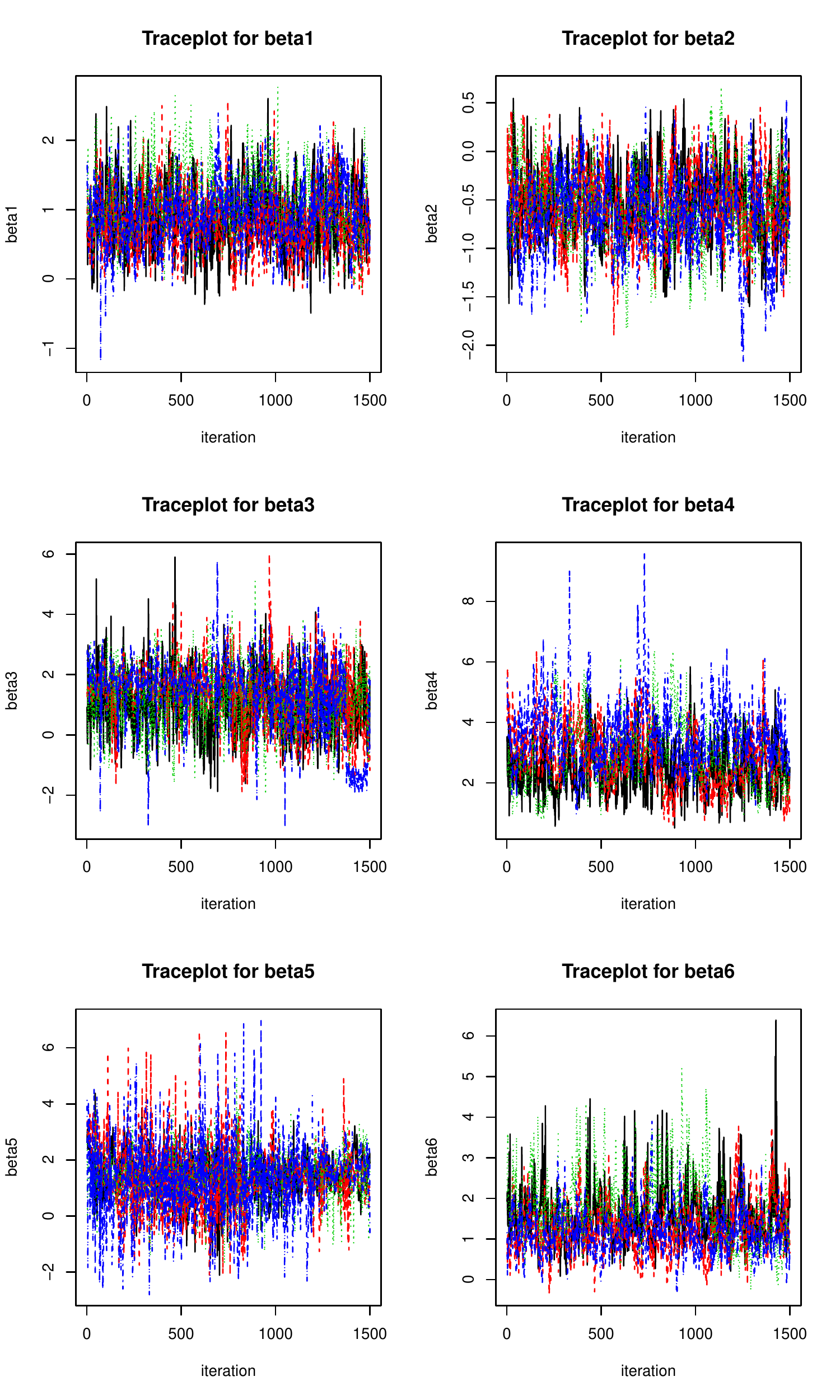}
    \caption{$\beta$ trace}
\end{figure}

\begin{figure}
    \centering
    \includegraphics[width=.75\textwidth]{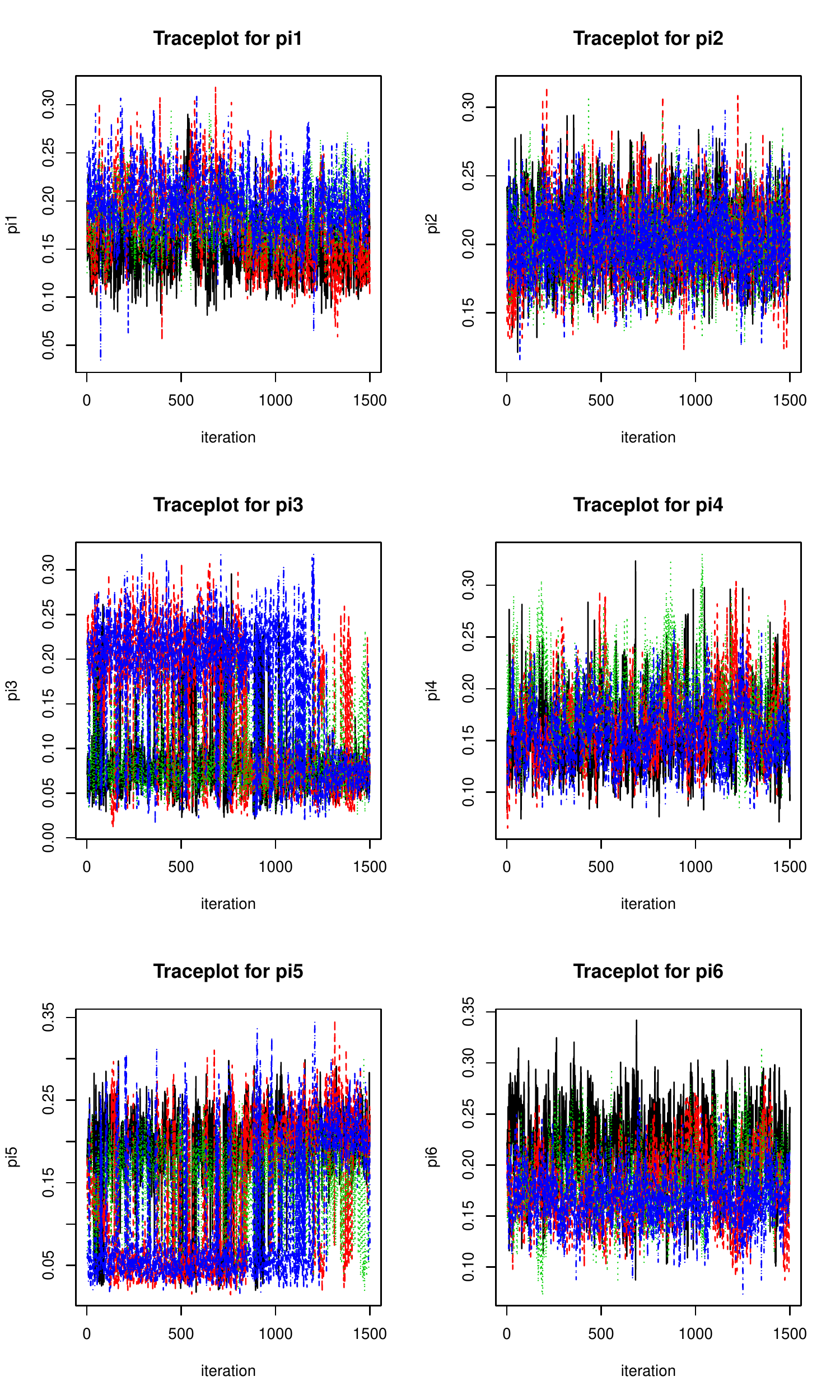}
    \caption{$\boldsymbol{\pi}$ trace}
\end{figure}

\begin{figure}
    \centering
    \includegraphics[width=.75\textwidth]{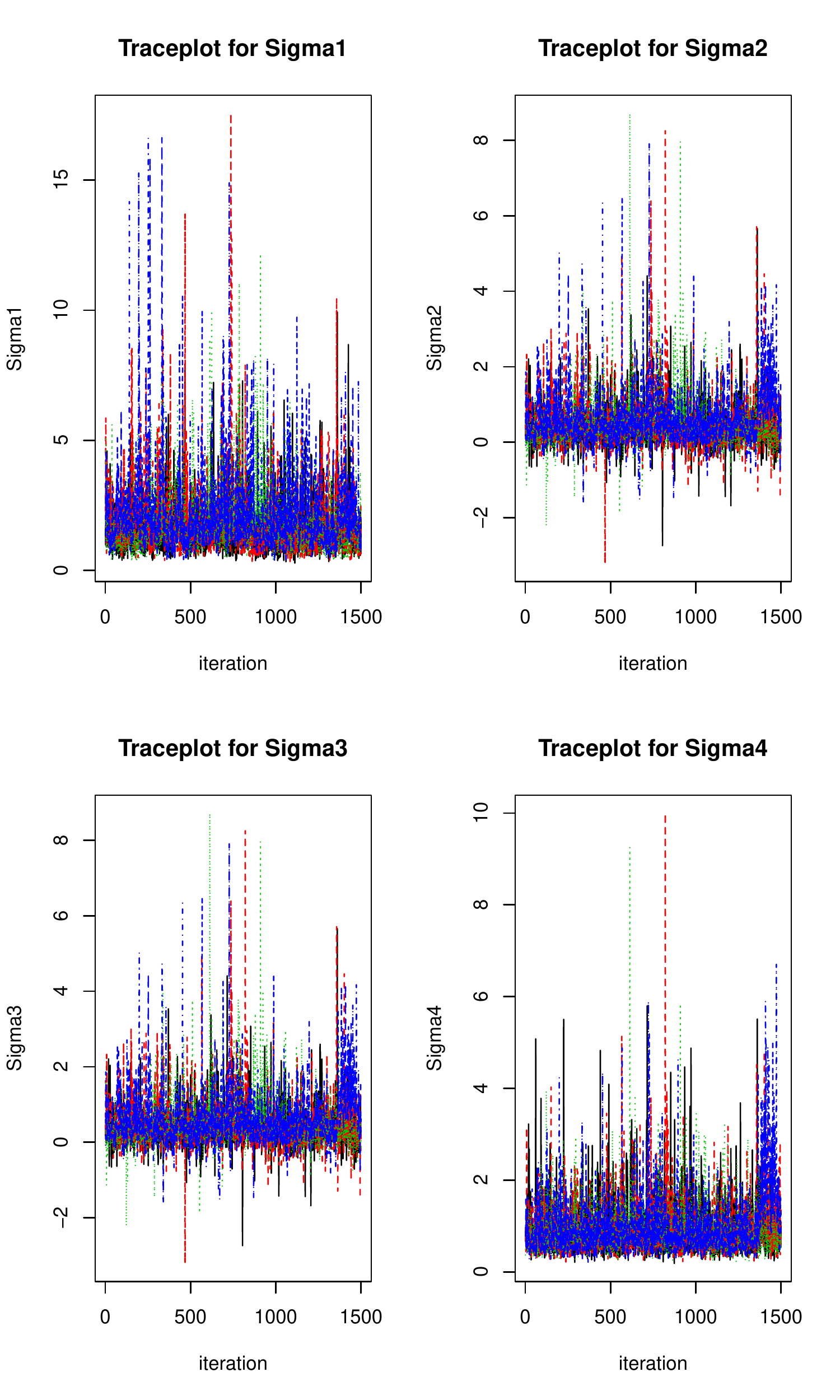}
    \caption{$\boldsymbol{\Sigma}$ trace}
\end{figure}
\newpage
\section{Simulation study details}

Each of the one thousand simulations in our simulation study generated a binary undirected network from the LS-SBM model on three hundred nodes and five equal-sized blocks as follows. The between-block probability matrix $B$ was:
\[ \mathbf{B} = \begin{pmatrix}- & 0.3 & 0.3 & 0.3 & 0.3 \\
0.3 & - & 0.3 & 0.3 & 0.3 \\
0.3 & 0.3 & - & 0.3 & 0.3 \\
0.3 & 0.3 & 0.3 & - & 0.3 \\
0.3 & 0.3 & 0.3 & 0.3 & - \end{pmatrix}. \]
Note that this between-block probability matrix cannot be  represented in a two-dimensional Euclidean latent space since it is impossible to place five points equidistant apart. The within-block two-dimensional latent space intercept $\boldsymbol\beta$ and scale $\boldsymbol\sigma$ were:
\begin{align*}
    \boldsymbol\beta &= (0.6, 2.0, 2.1, 4.0, 4.0) \\
\boldsymbol\sigma &= (0.4, 0.8, 1.2, 1.6, 2.0).
\end{align*} 

\section{Approximate computation algorithm}
\subsection{Two-stage approximate algorithm}

The approximate estimation occurs in two stages: 1) the stochastic blockmodel, and 2) the within-block latent spaces.

The goal of the first stage of the approximate algorithm is to quickly partition the nodes in to assortative clusters as in a stochastic blockmodel. Any graph clustering algorithm can in principle be used for this step. For small- and medium-sized networks (e.g.\ under 1,000 nodes), we have found a variant of spectral clustering designed to find assortative clusters to work well~\citep{saade2014spectral}. This is the algorithm we used in the analysis of the single Karnataka village and the simulation studies presented in the main text. For large networks, we have found label propagation~\citep{raghavan2007near} to scale well while still returning meaningful assortative clusters. This is the algorithm we used to estimate our model on all 75 Karnataka village networks combined (a total of 13,009 nodes).

For the second stage of the two-stage approximation, we estimate each within-block latent space for the blocks identified by spectral clustering using a variational Bayes algorithm. See appendix \ref{app:var_alg} for the details of this algorithm.

\subsection{Variational approximation}
 \label{app:var_alg}

For the variational Bayes estimation of the latent space in the two-stage approximate procedure, we adapt the algorithm developed in \cite{salter2013variational}.  Since the latent space associated with each block is estimated independently, we describe the algorithm below for a single block and omit the block subscript $k$ on the block-level parameters $\beta$ and $\sigma^2$.

Let $\tau = \sigma^{-2}$. For a prior distribution we set $\tau \sim \n{Gamma}(a_0, b_0)$ and $\beta \sim N(m_0, t_0^{-1})$. We use a fixed-form variational approximation with variational family of posterior distributions defined by  $\tau \sim \n{Gamma}(a, b)$, $\beta \sim N(m, t^{-1})$, $\mb{Z}_i \sim N_D(\ell_i, s_i^{-1} I_D)$, and $\tau, \beta, \mb{Z}_1, \dotsc, \mb{Z}_n$ independent. Denote the full set of free variational parameters $\psi$.

The standard variational criterion function, known as the ELBO, is not available in closed form for this approximation. We use the same first order Taylor series approximation to the part of the ELBO concerned with the likelihood of $Y_{ij}$ as in \cite{salter2013variational}:
\begin{align*}
\sum_{i,j} E_{\psi}\left[ \log p(Y_{ij} | \mb{Z}_i, \mb{Z}_j, \beta) \right] &\approx \sum_{i,j} \left[Y_{ij} \eta_{ij} - \log( 1+ \exp(\eta_{ij})) \right] 
\end{align*}
for
\[\eta_{ij} = m + t^{-1}/2 - (\|\ell_i - \ell_j\|^2 + (s_i^{-1} + s_j^{-1})d)^{1/2}.\]
The next part of the ELBO concerns the conditional distribution of $Z_i$:
\begin{align*}
\sum_i E_{\psi}\left[ \log \frac{p(\mb{Z}_i | \tau)}{q(\mb{Z}_i | \ell_i, s_i)}\right] &\propto \sum_i E_{\eta}\left[ \tfrac{1}{2}\log \tau  - \tfrac{\tau}{2N} \|\mb{Z}_i\|^2 - \tfrac{1}{2}\log s_i + \tfrac{s_i}{2}\| \mb{Z}_i - \ell_i\|^2\right]  \nonumber\\
&\propto  \sum_i \left[ \tfrac{1}{2}\psi(a) - \tfrac{1}{2}\log b - \tfrac{a}{2Nb}(ds_i^{-1} + \|\ell_i\|^2)- \tfrac{1}{2}\log s_i)\right].
\end{align*}
Finally, the KL divergence between prior and posterior:
\begin{align*}
E_{\eta}\left[ \log \frac{ p(\tau | a_0, b_0)}{q(\tau | a, b)} + \log \frac{p(\beta | m_0, t_0)}{q(\beta | m, t)} \right] &= E_{\eta}\left[ a_0\log b_0 -a\log b - \log \Gamma(a_0) + \log\Gamma(a)  \right. \nonumber \\
&\qquad \left. + (a_0 - a) \log \tau - (b_0 - b) \tau +\tfrac{1}{2} \log t_0 - \tfrac{1}{2}\log t \right. \nonumber \\
&\qquad \left.  -\tfrac{t_0}{2} (\beta - m_0)^2 + \tfrac{t}{2}(\beta - m)^2 \right]\nonumber \\
&\propto -a \log b + \log \Gamma(a) + (a_0 - a)(\psi(a) - \log b) \nonumber \\
&\qquad  - (b_0 - b) \tfrac{a}{b} - \tfrac{1}{2} \log t -\tfrac{t_0}{2} [(m - m_0)^2 + t^{-1}].
\end{align*}
We use a BFGS algorithm to maximize the ELBO with respect to $m, t, a, b, \ell$ and $s$. Denoting the ELBO $L$ and differentiating with respect to $m$ and $t$ gives:
\begin{align*}
\frac{\partial L}{\partial m}&= \sum_{i,j}\left[ Y_{ij} - \n{logit}^{-1}(\eta_{ij})\right] -t_0(m - m_0)\\
\frac{\partial L}{\partial t} &= \frac{1}{2t^2}\left(t_0+ \sum_{i,j}\n{logit}^{-1}(\eta_{ij})\right) - \frac{1}{2t}.
\end{align*} 
For $a$ and $b$ we have the derivatives
\begin{align*}
 \frac{\partial L}{\partial a}&= \psi'(a)(a_0 +\tfrac{N}{2} - a ) - b^{-1}\left(b_0 + \tfrac{1}{2N}\Sigma_i(\tfrac{D}{s_i} + \|\ell_i\|^2)\right) + 1. \\
 \frac{\partial L}{\partial b} &= -(a_0 + \tfrac{N}{2})b^{-1} +(a b_0 + \tfrac{1}{2N}\Sigma_i a(ds_i^{-1} + \|\ell_i\|^2))b^{-2}
\end{align*}
from which we get the closed form solution $\hat{a} = a_0 + \tfrac{N}{2}$. Finally for $\ell_i$ and $s_i$ we have
 \begin{align*}
  \frac{\partial L}{\partial \ell_{ik}} &= \sum_j \left[ -(\ell_{ik} - \ell_{jk})(\|\ell_i - \ell_j\|^2 + D(s_i^{-1} + s_j^{-1}))^{-1/2} \left(Y_{ij} - \n{logit}^{-1}(\eta_{ij})\right)\right] - \tfrac{a}{Nb}\ell_{ik} \\
\frac{\partial L}{\partial s_i} &= \sum_j \left[\tfrac{1}{2}Ds_i^{-2}(\|\ell_i - \ell_j\|^2 + D(s_i^{-1} + s_j^{-1}))^{-1/2} \left(Y_{ij} - \n{logit}^{-1}(\eta_{ij})\right)\right]  + \frac{ad}{2Nbs_i^2} - \frac{1}{2s_i}.
  \end{align*}

\end{document}